\documentclass[runningheads,a4paper]{llncs}
\usepackage{graphicx}
\newcommand\hmmax{0}
\newcommand\bmmax{0}

\gdef\dash---{\thinspace---\hskip.16667em\relax}

\newcommand{\newmodel}{\textsf{CompHO}}

\usepackage{pgf,tikz}
\usetikzlibrary{arrows,automata,shapes.arrows,positioning}

\usepackage{mathtools}
\usepackage{bm}
\usepackage{esvect}
\usepackage[utf8]{inputenc}
\usepackage{hyperref}
\usepackage{listings}          
\usepackage{amssymb}
\usepackage{stmaryrd}
\usepackage{graphicx}
\usepackage{xspace}
\usepackage{wrapfig}
\usepackage{subfig}
\DeclareCaptionType{copyrightbox}
\usepackage{xcolor}
\usepackage{newunicodechar}
\usepackage{mathpartir}
\usepackage[]{algorithm2e}

\usepackage[T1]{fontenc}
\usepackage{xcolor}
\usepackage{enumitem}
\usepackage{lipsum}
\usepackage{mathtools}
\usepackage{enumitem}
\usepackage{kantlipsum}

\hypersetup{
  unicode=true,    
  colorlinks=true,  
  citecolor=black,        
  urlcolor=blue           
}

\newcommand{\seqcomp}[1]{\mathrm{mono}(#1)}

\newcommand{\mytilde}{\kern -.15em\lower .7ex\hbox{\~{}}\kern .04em}

\newcommand{\ignore}[1]{}

\newcommand{\ho}{\ensuremath{\mathit{HO}}\xspace}

\newcommand{\algorw}{{\texttt{make-CompHO}}\xspace}

\renewcommand{\next}{\textit{succ}}

\newcommand{\ph}{\ensuremath{\mathit{ph}}}
\newcommand{\rd}{\ensuremath{\mathit{rd}}}

\newcommand{\aexec}{\ensuremath{\pi}}

\newcommand{\ra}{\ensuremath{\rightarrow}}
\newcommand{\prj}[1]{\hspace{-1.4mm}\downharpoonright_{#1}}

\newcommand{\asyncexec}[1]{\ensuremath{\textsf{ae}^{#1}}}
\newcommand{\syncexec}[1]{\ensuremath{\textsf{se}^{#1}}}

\newcommand{\tuple}[1]{\left<{#1}\right>}

\newcommand{\prog}{\ensuremath{\mathcal{P}}\xspace}

\newcommand{\nothing}{\ensuremath{\ast}\xspace}
\newcommand{\init}{\text{\texttt{init}}\xspace}

\newcommand{\send}{\text{\texttt{send}}\xspace}
\newcommand{\phase}{\text{\texttt{phase}}\xspace}

\newcommand{\update}{\text{\texttt{update}}\xspace}

\renewcommand{\path}{\ensuremath{\mathtt{\pi}}\xspace}

\newcommand{\pc}{\ensuremath{\mathsf{pc}}}

\newcommand{\loc}{\ensuremath{\mathtt{Loc}}\xspace}
\newcommand{\vars}{\ensuremath{\mathtt{Vars}}\xspace}

\newcommand{\svars}{\ensuremath{\mathtt{SyncV}}\xspace}

\newcommand{\tags}{\ensuremath{\mathtt{tags}}\xspace}
\newcommand{\tagm}{\ensuremath{\mathtt{tagm}}\xspace}
\newcommand{\dompaytype}{\ensuremath{\mathcal{T}}\xspace}
\newcommand{\paytype}{\ensuremath{\mathbb{MT}}\xspace}
\newcommand{\pidtype}{\ensuremath{\mathsf{Pid}}\xspace}
\newcommand{\type}{\ensuremath{\mathbb{T}}\xspace}

\newcommand{\datatype}{\ensuremath{\mathcal{D}}\xspace}

\newcommand{\limp}{\Rightarrow} 

\newcommand{\sem}{[\![}
\newcommand{\antic}{]\!]}

\newunicodechar{∃}{\ensuremath{\exists}}
\newunicodechar{∀}{\ensuremath{\forall}}
\newunicodechar{θ}{\ensuremath{\theta}}
\newunicodechar{τ}{\ensuremath{\tau}}
\newunicodechar{φ}{\ensuremath{\varphi}}
\newunicodechar{π}{\ensuremath{\pi}}
\newunicodechar{α}{\ensuremath{\alpha}}
\newunicodechar{β}{\ensuremath{\beta}}
\newunicodechar{γ}{\ensuremath{\gamma}}
\newunicodechar{δ}{\ensuremath{\delta}}
\newunicodechar{ε}{\ensuremath{\epsilon}}
\newunicodechar{λ}{\ensuremath{\lambda}}
\newunicodechar{ρ}{\ensuremath{\rho}}
\newunicodechar{σ}{\ensuremath{\sigma}}
\newunicodechar{ω}{\ensuremath{\omega}}
\newunicodechar{Γ}{\ensuremath{\Gamma}}
\newunicodechar{Φ}{\ensuremath{\Phi}}
\newunicodechar{Δ}{\ensuremath{\Delta}}
\newunicodechar{Σ}{\ensuremath{\Sigma}}
\newunicodechar{Π}{\ensuremath{\Pi}}
\newunicodechar{Θ}{\ensuremath{\Theta}}
\newunicodechar{⇒}{\ensuremath{\Rightarrow}}
\newunicodechar{⇐}{\ensuremath{\Leftarrow}}
\newunicodechar{⇔}{\ensuremath{\Leftrightarrow}}
\newunicodechar{→}{\ensuremath{\rightarrow}}
\newunicodechar{←}{\ensuremath{\leftarrow}}
\newunicodechar{↔}{\ensuremath{\leftrightarrow}}
\newunicodechar{¬}{\ensuremath{\neg}}
\newunicodechar{∧}{\ensuremath{\land}}
\newunicodechar{∨}{\ensuremath{\lor}}
\newunicodechar{≠}{\ensuremath{\neq}}
\newunicodechar{≡}{\ensuremath{\equiv}}
\newunicodechar{∼}{\ensuremath{\sim}}
\newunicodechar{≈}{\ensuremath{\approx}}
\newunicodechar{≥}{\ensuremath{\geq}}
\newunicodechar{≤}{\ensuremath{\leq}}
\newunicodechar{≫}{\ensuremath{\gg}}
\newunicodechar{≪}{\ensuremath{\ll}}
\newunicodechar{∅}{\ensuremath{\emptyset}}
\newunicodechar{⊆}{\ensuremath{\subseteq}}
\newunicodechar{⊂}{\ensuremath{\subset}}
\newunicodechar{∩}{\ensuremath{\cap}}
\newunicodechar{∪}{\ensuremath{\cup}}
\newunicodechar{∈}{\ensuremath{\in}}
\newunicodechar{⊤}{\ensuremath{\top}}
\newunicodechar{⊥}{\ensuremath{\bot}}
\newunicodechar{₀}{\ensuremath{_0}}
\newunicodechar{₁}{\ensuremath{_1}}
\newunicodechar{₂}{\ensuremath{_2}}
\newunicodechar{₃}{\ensuremath{_3}}
\newunicodechar{₄}{\ensuremath{_4}}
\newunicodechar{₅}{\ensuremath{_5}}
\newunicodechar{₆}{\ensuremath{_6}}
\newunicodechar{₇}{\ensuremath{_7}}
\newunicodechar{₈}{\ensuremath{_8}}
\newunicodechar{₉}{\ensuremath{_9}}
\newunicodechar{⁰}{\ensuremath{^0}}
\newunicodechar{¹}{\ensuremath{^1}}
\newunicodechar{²}{\ensuremath{^2}}
\newunicodechar{³}{\ensuremath{^3}}
\newunicodechar{⁴}{\ensuremath{^4}}
\newunicodechar{⁵}{\ensuremath{^5}}
\newunicodechar{⁶}{\ensuremath{^6}}
\newunicodechar{⁷}{\ensuremath{^7}}
\newunicodechar{⁸}{\ensuremath{^8}}
\newunicodechar{⁹}{\ensuremath{^9}}
\newunicodechar{ℝ}{\ensuremath{\mathbb{R}}}
\newunicodechar{ℕ}{\ensuremath{\mathbb{N}}}
\newunicodechar{ℂ}{\ensuremath{\mathbb{C}}}
\newunicodechar{ℚ}{\ensuremath{\mathbb{Q}}}
\newunicodechar{ℤ}{\ensuremath{\mathbb{Z}}}
\newunicodechar{✓}{\checkmark}
\newunicodechar{✗}{\ensuremath{\times}}
\newunicodechar{◊}{\ensuremath{\lozenge}}
\newunicodechar{□}{\ensuremath{\square}}

\lstdefinelanguage{scala}{
  morekeywords={abstract,case,catch,class,def,%
    do,else,extends,false,final,finally,%
    for,if,implicit,import,match,mixin,%
    new,null,object,override,package,%
    private,protected,requires,return,sealed,%
    super,this,throw,trait,true,try,%
    type,val,var,while,with,yield,%
    Round,ProcessID,Int,send,update,init,old,process,Boolean,Set,Map,variable,interface,receive},
  otherkeywords={=>,<-,<\%,<:,>:,\#,@},
  sensitive=true,
  morecomment=[l]{//},
  morecomment=[n]{/*}{*/},
  morestring=[b]",
  morestring=[b]',
  morestring=[b]"""
}

\lstdefinelanguage{ho}{
  morekeywords={abstract,case,catch,class,def,%
    typedef, struct,%
    int, bool,%
    do,if, else,extends,false,final,finally, continue,break,%
    for,implicit,import,match,mixin,%
    new,null,object,override,package,%
    private,protected,requires,return,sealed,%
    super,this,throw,trait,true,try,%
    type,val,var,while,with,yield,%
    Round,ProcessID,Int,send,update,init,old,process,Boolean,Set,Map,variable,interface,receive, round, UPDATE, SEND, out_internal,exit, HOmachine, in,out},
  otherkeywords={=>,<-,<\%,<:,>:,\#,@},
  sensitive=true,
  morecomment=[l]{//},
  morecomment=[n]{/*}{*/},
  morestring=[b]",
  morestring=[b]',
  morestring=[b]"""
}

\definecolor{dkgreen}{rgb}{0,0.6,0}
\definecolor{gray}{rgb}{0.5,0.5,0.5}
\definecolor{mauve}{rgb}{0.58,0,0.82}

\def\ContinueLineNumber{\lstset{firstnumber=last}}

\newcommand{\citet}[1]{\cite{#1}}

\begin{document}

\title{Communication-closed asynchronous protocols}
%
%
\author{Andrei Damian\inst{3} \and
Cezara Dr\u{a}goi\inst{1}\and Alexandru Militaru
\inst{3}\and Josef Widder\inst{2}}
%
%
\institute{INRIA, ENS, CNRS, PSL \and
TU Wien \and Politehnica University Bucharest}
\maketitle              
\begin{abstract}
Fault-tolerant distributed systems are implemented over asynchronous
     networks, so that they use algorithms for  asynchronous models
     with faults.
Due to asynchronous communication and the occurrence of faults (e.g.,
     process crashes or the network dropping messages)  the
     implementations are hard to understand and analyze.
In contrast, synchronous computation models simplify design and
     reasoning.
In this paper, we bridge the gap between these two worlds.
For a class of asynchronous protocols, we introduce a procedure that,
     given an asynchronous protocol, soundly computes its round-based
     synchronous counterpart.
This class is defined by properties of the sequential code.

We computed the synchronous counterpart of known consensus and leader
     election protocols, such as, Paxos, and Chandra and Toueg's
     consensus.
Using Verifast we checked the sequential properties required by the
rewriting.
We verified the round-based synchronous counter-part of Multi-Paxos, and other algorithms, using
existing deductive verification methods for synchronous protocols.
\end{abstract}

%
%
%
\section{Introduction}\label{sec:intro}

\newcommand{\impcite}{\cite{Chandra2007PML,OngaroO14,HawblitzelHKLPR17,EPaxos,RenesseSS15}}


Fault-tolerant distributed systems provide a dependable service on top
     of unreliable computers and networks.
They implement fault tolerance protocols that replicate the system and ensure that
      from the outside all (unreliable) replicas are perceived  as a
     single reliable one.
This has been formalized by strong
     consistency, consensus, state machine replication. 
These protocols are crucial parts of many distributed systems and their
     correctness is very hard to obtain.
Protocol designers are faced with the challenges of buffered message
     queues, message re-ordering at the network, message loss,
     asynchrony and concurrency,
     and process faults.
Reasoning about all these features is a notoriously hard as
     discussed in several research papers~\impcite, and as a
     consequence testing tools like Jepsen~\cite{Jepsen} found conceptual design
     bugs in deployed implementations.

\newcommand{\ITP}{\cite{WoosWATEA16,LesaniBC16,RahliGBC15,HawblitzelHKLPR17,PadonMPSS16,PadonLSS17,atomicGotsman}}

\newcommand{\IMPL}{\cite{Zab,M12zoo,OngaroO14,EPaxos,Kotla2010}}

\paragraph{Problem statement.}
A programming abstraction of synchronous
     rounds \cite{DLS88:jacm,Lyn96:book,Charron-BostS09}  would
     relieve the designer from many of these difficulties.
Synchronous round-based algorithms are more structured, are easier to
     understand, have simpler behaviors.
As one only has to reason about specific global states at the round
     boundaries, they entail simpler correctness  arguments.
However, it is also well-understood that  synchronous distributed systems,
     are often ``impossible or inefficient to
     implement''~\cite[p.~5]{Lyn96:book}.
Hence, designers turn to the asynchronous model, in which the
     performance emerges~\cite{Lann03} from the current load of a
     system, which in normal operation has significant better
     performance.
Thus, no synchronous algorithm is used in any real large scale system
     we are aware of.

In face of the different advantages of synchronous and asynchronous
     models, the question is how to connect these two worlds.
We consider the question, given an asynchronously
     algorithm, does it have a synchronous canonic counterpart?

\paragraph{Key Challenges.}
The main difficulty stems from the fundamental discrepancy in the
control structure, that is,  \emph{(i)}
interleavings and \emph{(ii)} message buffers:

\noindent
\emph{(i)
Interleavings:} In  synchronous round-based models, computation
     is structured in rounds that  are executed by all process in lock-step.
There is no interleaving between steps, the beginning and the end of
each step is synchronized across processes, by definition.
In the asynchronous computational model executions are much less
     structured.
Processes are scheduled according to interleaving semantics.
This leads to an exponential number of intermediate
     global states (exponential in the number of steps)
     vs.\ a linear one in the synchronous case.

     \smallskip

\noindent
\emph{(ii) Message buffers:} In the synchronous model, messages are
     received in the same round they are sent.
Thus, the number of messages that are in-transit is bounded at all
     times, and depends on the number of processes.
In the asynchronous model, fast processes may
     generate messages quicker than slow processes may process them.
Thus communication needs to be buffered, and the buffer size is
     unbounded.
The number of messages that are in a buffer depends on the
     number of processes, but also on the number of send instructions
     executed by each process.
Moreover, the network may reorder messages, that is,
      a process may receive a new message before all older ones,
     that are still in-transit.

    \smallskip

Due to the discrepancy, there is no obvious reason why an
     asynchronous algorithm should have a synchronous
     canonic form.
In general there is none.
We characterize asynchronous systems that allow to
     dissolve this discrepancy.

     \paragraph{Key Insights.}
As not all conceivable asynchronous protocols can be rewritten into
     synchronous ones, we focus on characteristics of practical
     distributed systems.
From a high-level viewpoint, distributed systems are about
     coordination in the absence of a global clock.
Thus, distributed algorithms implement an abstract notion of time to
     coordinate.
This notion of time may be implicit.
However,  the local state of a process maintains  this abstract time
     notion, and a process timestamps the messages it sends
     accordingly.
Synchronous algorithms do not need to implement an abstract notion of
     time, as it is present from the beginning: the round number plays
     this role and it is embedded in the definition of any synchronous
     computational model.
The key insight of our results is the existence of a correspondence
     between values of the abstract clock in the asynchronous systems
     and round numbers in the synchronous ones.
Using this correspondence, we  make explicit the
     ``hidden'' round-based synchronous structure of an asynchronous
     algorithm.
More systematically, in an asynchronous system:
\begin{itemize}
\item abstract time is encoded in local variables.
     Modifications of their values mark the progress
     of abstract time, and\dash---making the correspondence to
     round-based algorithms\dash---the local beginning/ending of
     a round;

\item a global round consists of all the steps processes execute
  for the same value of the abstract clock;

\item messages are timestamped with the value of the abstract time of the
     sender, when the message was sent; a receiver can read the
     abstract time at which the  message was sent, and compare
     it with its own abstract time;

\item in order to have a faithful correspondence to
  round-based semantics, we consider \emph{communication-closed protocols}:
  the reception of a message is effectful only if its timestamp is
  equal to or greater than the local time of the receiver. In other
  words, stale messages are discarded.

\end{itemize}

Based on these insights (i) we characterize asynchronous protocols
     whose executions can be reduced to well-formed canonic
     executions, (ii) we define a computational model \newmodel\ whose
     executions are all canonic by definition, (iii) we show how to
     translate an asynchronous protocol into code for the \newmodel\
     framework, and (iv) we show the benefits of this computed canonic
     form for design, validation, and verification of distributed
     systems.
We discuss these four points by using a running example in the
     following section.


\section{Our approach at a glance}
\label{sec:glance}

\begin{figure}[t]
\begin{minipage}{0.49\textwidth}
\begin{lstlisting}[frame=none, language=C]
typedef struct Msg {  int lab; int ballot; int sender;} msg;
typedef struct List{  msg *message;	struct List * next; int size;} list; <@\label{ln:listtype}@>
enum labels {NewBallot, AckBallot} ;
int coord(); <@\label{ln:leader}@>
bool filter(msg *m,int b,enum labels l){ <@\label{ln:filter}@>
  if (m!=0 && l==NewBallot)  return (m->ballot>=b&& m->label==l);<@\label{ln:filterge}@>
  if (m!=0 && l==AckBallot) return (m->ballot==b  && m->label==l);
  return false; }
bool all_same(list *mbox){
list  *x=mbox;
if (x!=NULL) val = x->message->sender;
while(x!=NULL) { if(x->message->sender != val) return false;  else x= x->next;}
return true;}
int main(){
int  me=getId(); int n;
struct arraylist *log_epoch, *log_leader;
int ballot = 0;  enum labels label;
list *mbox = NULL; list_create(log_epoch);
list_create(log_leader);
while(true){
 label = NewBallot; <@\label{ln:origstartNewBallot}@>
 if (coord() == me){//LEADER'S CODE <@\label{ln:checkleader}@>
   ballot++;
   //@ assert  tag_leq(oldballot, oldlabel, ballot, label);<@\label{ln:tagleq}@>
   msg *m = create_msg(ballot,label,me);
   //@ assert tag_eq(m->ballot,m->label, ballot,label);<@\label{ln:sendassert}@>
   send(m,*);  <@\label{ln:sendelectleader}@> //send new ballot to all
   reset_timeout();
   while(true ){ <@\label{ln:startrcv}@> (m,p) = recv(); <@\label{ln:rcvelect}@>
    if (filter(m,ballot,label)) add(mbox, m); <@\label{ln:filterrcvelect}@>
    if((mbox!=0 && mbox->size==1 )||timeout())break;} <@\label{ln:endrcv}@>
//@ assert  mbox_geq(ballot,label,mbox); <@\label{ln:msgtag}\label{ln:recvassert} @>
   if (mbox!=0 && mbox->size ==1) {<@\label{ln:context}@>
        ballot = mbox->message->ballot;  <@\label{ln:setballot}@>
     //@assert max_tag(ballot,label,mbox);
       leader = mbox->message->sender;<@\label{ln:setleader}@>
     \end{lstlisting}
	\end{minipage}
\begin{minipage}{0.49\textwidth}
  \ContinueLineNumber
\begin{lstlisting}[frame=none, language=C]
dispose(mbox); label = AckBallot; <@\label{ln:origstartAckBallotL}@>
//@ assert  tag_leg(oldballot, oldlabel, ballot,label);
msg *m = create_msg(ballot,label,me);
//@ assert tag_eq(m->ballot, m->label, ballot,label);
send(m,*);  <@\label{ln:sendelect}@> //send ack ballot to all
reset_timeout(); //wait for AckBallot
while  (true){ (m,p) = recv();  <@\label{ln:rcvleader}@>
    if (filter(m,ballot,label)) add(mbox, m);
    if((mbox!=0 && mbox->size>n/2 ) break;
    if(timeout()) break;}<@\label{ln:endrcv2}@>
//@ assert  mbox_tag_eq(ballot,label,mbox); <@\label{ln:listtageq}@>
if (mbox!=0 && mbox->size > n/2&& all_same(mbox)){
    list_add(log_ballot,ballot, true); <@\label{ln:settrue}@>
    list_add(log_leader,ballot, leader); <@\label{ln:elect}@>
    out(ballot, leader);}
}else{  dispose(mbox); }
}else{//FOLLOWER'S CODE
  ballot++; reset_timeout(); //wait Newballot
  while(true ){ <@\label{ln:folstartrcv}@>  (m,p) = recv(); <@\label{ln:rcvelectF}@>
     if (filter(m,ballot,label)) add(mbox, m);
     if((mbox!=0 && mbox->size==1 )||timeout())break;} <@\label{ln:folendrcv}@>
  if (mbox!=0 && mbox->size ==1) {<@\label{ln:contextF}@>
    ballot = mbox->message->ballot; <@\label{ln:folsetballot}@>
    leader = mbox->message->sender; {<@\label{ln:folsetleader}@>
    dispose(mbox); label = Ackballot; <@\label{ln:origstartAckBallotF}@>
    msg *m = create_msg(ballot, label, leader);
    send(m,*);  <@\label{ln:folsendelect}@> //send ack ballot to all
    reset_timeout(); //wait for Ackballot
    while(true){ (m,p) = recv(); <@\label{ln:rcvfollower}@>
      if (filter(m,ballot,label))add(mbox, m);
      if(mbox!=0 && mbox->size>n/2)break;
      if(timeout())break;}<@\label{ln:folendrcv2}@>
     if(!timeout() && all_same(mbox)){
	 list_add(log_ballot,ballot, true); <@\label{ln:folsettrue}@>
	 list_add(log_leader,ballot, leader);
	 out(ballot, leader); }<@\label{ln:folelect}@>
      }else{  dispose(mbox); }
}}} //END FOLLOWER END while END main
\end{lstlisting}
\end{minipage}
\caption{Leader election protocol inspired from Paxos (Phase 1)}
\label{fig:async-leader-election}
\vspace{-3eX}
\end{figure}

The running example in Fig~\ref{fig:async-leader-election} is inspired
     by typical fault-tolerant distributed protocols that often rely
     on the notion of leadership.
For example, primary back-up algorithms use a leader to order the
     client requests and to ensure this order among all replicas.
A leader is a process that is connected via timely links to a majority
     of replicas.
Hence, only in ballots where such a well-connected leader exists, the
     system should try to make progress.
The algorithm in Fig.~\ref{fig:async-leader-election} implements just
     the leader election algorithm.
All processes execute the same code, and $n$ is the number of
     processes.
In each loop iteration each process queries its \texttt{coord} oracle
     in line~\ref{ln:checkleader} to check whether it is a leader
     candidate.
Multiple processes may be candidates in the same iteration.
Depending on the outcome, the code then branches to a leader branch
     and a follower branch.
A candidate process sends to all in line~\ref{ln:sendelectleader}.
Then, the leader branch has the same code as the follower branch, that
     is, waiting for the first message by a candidate in the loop
     starting at lines~\ref{ln:startrcv} and~\ref{ln:folstartrcv}.
Thus, if there are multiple candidates there is a race between them.
If a message from a candidate for a current of future ballot is
     received, processes update their ballot in
     line~\ref{ln:setballot} and~\ref{ln:folsetballot}, and then set
     their leader estimate in the next line.
This estimate is then sent to all in line~\ref{ln:sendelect}
     and~\ref{ln:folsendelect}, and then processes wait to receive
     messages from a majority ($n/2$), for the current ballot.
If all $n/2$  received messages carry the same leader identity, then a
     process knows a leader is elected in the current ballot, and it
     records the leader's identity in line~\ref{ln:elect}
     and~\ref{ln:folelect}.
From a more structural viewpoint, because this protocol is supposed to
     be fault-tolerant, the receive statements in, e.g.,
     line~\ref{ln:rcvelect} and~\ref{ln:rcvleader} are non-blocking
     and may receive \texttt{NULL} if no message is there.

\begin{figure}[t]
  \centering
  \vspace{-3eX}
  \includegraphics[scale=0.30]{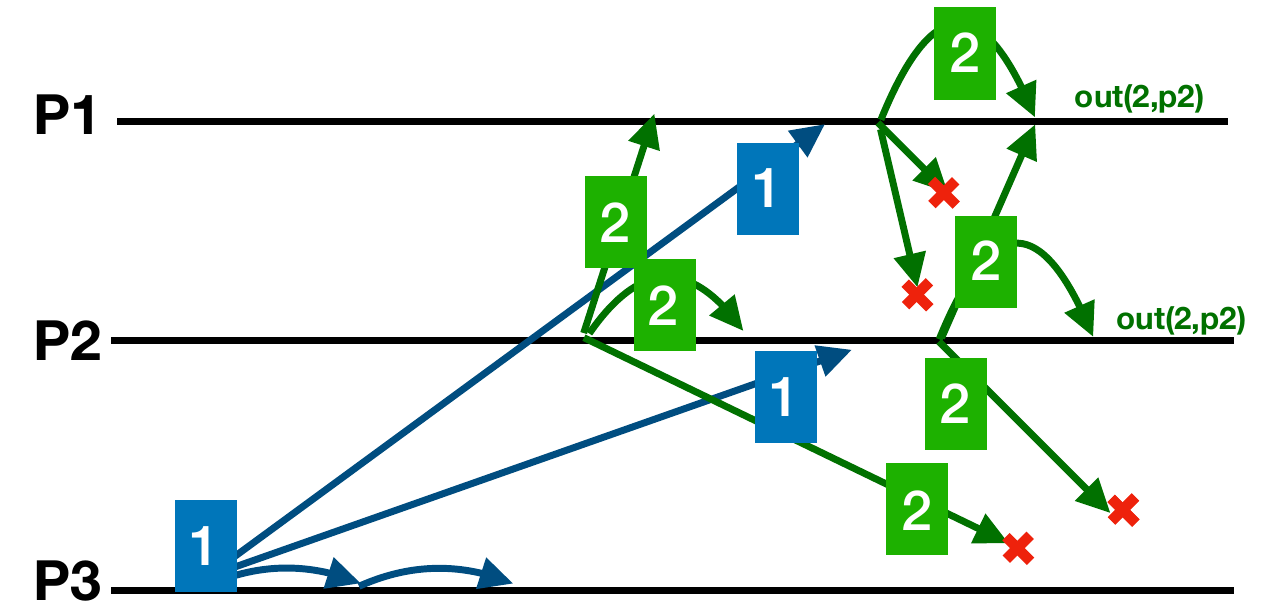}\hspace{1cm}
  \includegraphics[scale=0.30]{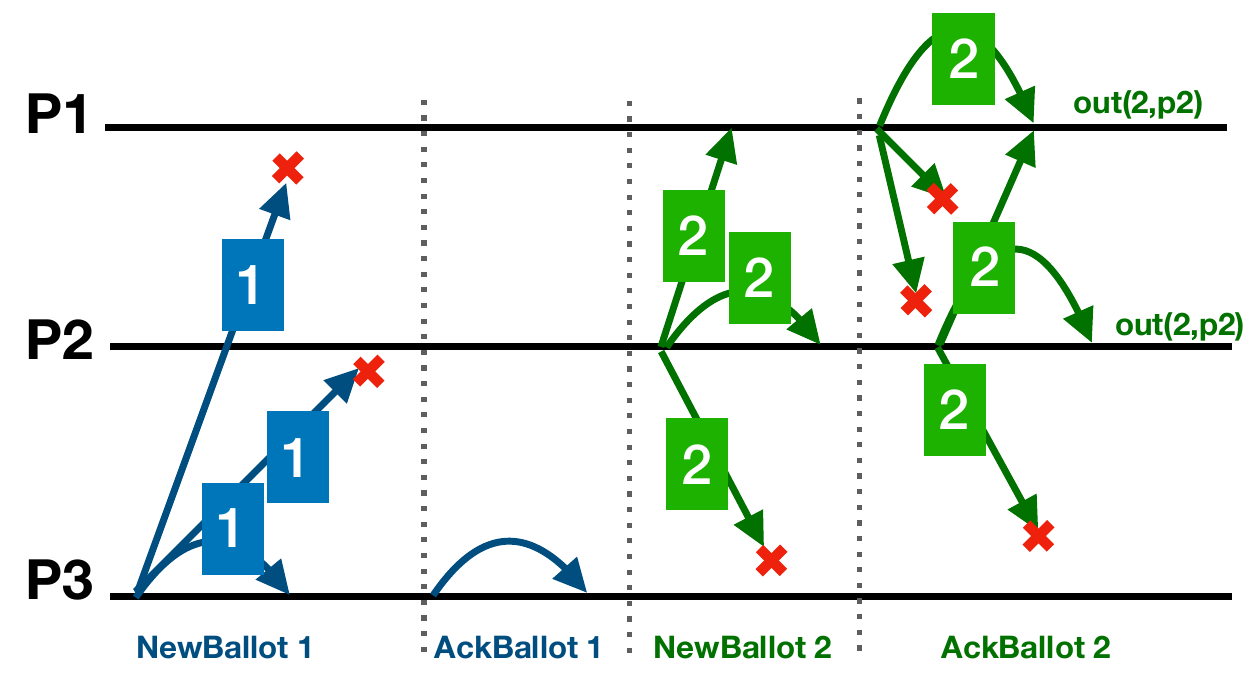}
\vspace{-1eX}
\caption{Asynchronous execution with jumps on the left and
  corresponding ``synchronous'' execution on the right}
\label{fig:asyncLE_Exec-jump1}
\end{figure}

Consider the asynchronous execution on the left of
     Figure~\ref{fig:asyncLE_Exec-jump1}.
Process~P1 always takes the follower branch, P3 is
     a candidate in ballot~1, but its messages get delayed.
So process~P2 times out in ballot~1 in line~\ref{ln:folendrcv}
     and~\ref{ln:folendrcv2}, and becomes a candidate in ballot~2.
Its message reaches P1, which jumps to ballot~2 and sends its leader
     estimate to all in line~\ref{ln:folsendelect}, as does P2.
However, only P1 receives all these message so that it gets over the
     $n/2$ threshold to set its leader in line~\ref{ln:folelect}.
Messages marked with a cross are dropped by the network.
As the messages by P1 arrive late, i.e., the receiver's local time passed the timestamp of the message, the messages sent by P1 become stale and are
     disregarded by P2 and~P3.

As a result, the late messages sent by P1 have the same effect as if they
     were dropped by the network.
In this view,  the execution of the right on
     Figure~\ref{fig:asyncLE_Exec-jump1} is obtained from the one  on
     the left by a ``rubber band transformation''~\cite{Mattern89},
     that is, transitions are reordered by maintaining to local
     control flow, and the causality imposed by sending and receiving
     a message.
We call the execution on the right the canonic form of the execution
     on the left.

\paragraph{Characterization of existence of a canonic form.}

Let us understand whether each execution of the asynchronous protocol
     from the example can be brought into a canonic form.
The first observation is that the variables \texttt{ballot} and
     \texttt{label} encode abstract time.
Let $b$ and $\ell$ be evaluations of the variables \texttt{ballot} and
     \texttt{label}.
Then abstract time ranges over $T =  \{(b,\ell)\colon b\in\mathbb{N},
     \ell\in
     \{\mbox{\texttt{NewBallot}},\mbox{\texttt{AckBallot}}\}\}$.
We fix $\mbox{\texttt{NewBallot}}$ to be less than
     $\mbox{\texttt{AckBallot}}$, and consider  the lexicographical
     order over $T$.
Then we observe that the sequence of $(b,\ell)$ induced by an
     execution at a process is monotonically increasing; thus
     $(b,\ell)$ encodes a notion of time.
However, a locally monotonic ascending sequence of values is not
     sufficient  to derive a global notion of time, i.e., a globally
     aligned ascending sequence of values, as in
     Figure~\ref{fig:asyncLE_Exec-jump1} on the right.
Technically, aligning means that we need a reduction argument where
     (i) we tag an event with the local time of the process at which
     it occurs, and (ii) if in an execution a transition $t$ tagged
     with $(b,\ell)$ happens before a transition $t'$ tagged with
     $(b',\ell')$ at another process, with $(b',\ell')<(b,\ell)$, then
     swapping these two transition should again result in an
     execution.
That is, $t$ and $t'$ should commute.
This condition is satisfied when stale messages are discarded.
In other words, it is ensured if the protocol is communication-closed:
     first, each process only sends for the current timestamp, e.g.,
     the send statement in line~\ref{ln:sendelect} sends a message
     that carries the current $(\texttt{ballot},\texttt{label})$ pair.
Second, each process receives only for the current or a higher
     timestamp, e.g., received messages are stored, e.g., in
     line~\ref{ln:filterrcvelect}, only if they carry the current or a
     future $(\texttt{ballot}, \texttt{label})$ pair;
     cf.~line~\ref{ln:filter}.

We introduce a tag annotation in which the programmer can provide us
     with the variables and parts of the messages that are supposed to
     encode abstract time and timestamps.
Using these tags, the protocol is annotated
     with verification conditions stating that (1) abstract time  is
     monotonically increasing (line~\ref{ln:tagleq}), which use the
     predicate $tag\_leq$ in Fig.~\ref{fig:preds}
     (page~\pageref{fig:preds}),   (2) each process sends only for the
     current timestamp (line~\ref{ln:sendassert}), and  (3) each
     process receives messages from the current or a higher timestamp
     (line~\ref{ln:recvassert}).
Given an annotated asynchronous protocol, we check the validity of
     these assertions  using the static verifier
     Verifast~\cite{verifast}.
Using Verifast was extremely useful to prove the conditions on the
     content of the mailbox,  which is an unbounded list
     (lines~\ref{ln:listtype} and~\ref{ln:listtageq}).
For example, the \texttt{assert} at line~\ref{ln:listtageq} states
     that all messages in the mailbox have their ballot and label
     fields equal with the  local variable \texttt{ballot} and
     \texttt{label}.
The other predicates are given Fig.~\ref{fig:preds}.

If these checks are successful, the existence of a canonical form is
     guaranteed by a new reduction theorem proven in
     Section~\ref{ssec:reduction}.
Our reduction uses ideas from~\cite{EF82} where the notion of
     communication closure was introduced in CSP.
Our notion of communication closure is more permissive than the
     original form, as we allow to react to  a message that is
     timestamped with higher value $(b',\ell')$ than the current local
     abstract time~$(b,\ell)$, provided that the code immediately
     ``jumps forward in time'' to $(b',\ell')$.
This corresponds in our example to P1 jumping to ballot~2 upon
     reception of the message by P2.

In contrast to Lipton's reduction~\cite{Lipton75}, where one proves
     that actions in an execution can be moved in order to get a
     similar execution with large atomic blocks of local code, for
     distributed algorithms one proves that one can group together the
     $(b,\ell)$ send transitions of all processes, then the~$(b,\ell)$
     receive transitions of all processes, and then all~$(b,\ell)$
     computation steps, for all times $(b,\ell)$ in increasing order.
In this way, we formally establish that the asynchronous execution
     from the left of Figure~\ref{fig:asyncLE_Exec-jump1} corresponds
     to the so-called round-based execution on the right.
Executions of this form we call \emph{canonic}.

\lstset{language=ho,numbers=left,xleftmargin=2em,frame=single,framexleftmargin=1.5em}
\begin{figure}[t]
\begin{minipage}{0.49\textwidth}
\begin{lstlisting}[frame=none]
typedef struct Msg {int lab; int ballot; int sender;} msg;			    
typedef struct List{  msg *message;	struct List * next; int size;} list; <@\label{ln:slisttype}@>
int coord(){return pid } <@\label{ln:sleader}@>
bool all_same(list *mbox){
  list  *x=mbox;
  if (x!=NULL) val = x->message->sender;
  while(x!=NULL) { if(x->message->sender != val) return false;  else x= x->next;}
 return true;}
int phase(){return round/phase.length}
int coord();
void init(){
    me = getMyId();
    old_mbox1 = false;
    list_create(log_epoch);
    list_create(log_leader);
}

   \end{lstlisting}
\end{minipage}
\begin{minipage}{0.49\textwidth}
  \ContinueLineNumber
\begin{lstlisting}[frame=none]
phase = array[NewBallot; AckBallot]
round NewBallot:
send{
     if(coord()==me)  send(pid,*); <@\label{ln:leadsend}@>}
update(list * mbox){<@\label{ln:updatenew}@>
    if (mbox!= 0 && mbox->size ==1) {<@\label{ln:condupd1}@>
	leader = mbox->message->sender; <@\label{ln:setleadersync}@>
	old_mbox1 = true;}<@\label{ln:newend}@>
round AckEpoch:
send{
	if (old_mbox1) send(leader,*);}
update(list* mbox){
  if(old_mbox1 == true && mbox!=0 && mbox->size >n/2 && all_same(mbox)) {<@\label{ln:ackstart}@>
      out(phase(),leader);
      list_add(log_epoch,phase(),true); <@\label{ln:settruesync}@>
      list_add(log_leader,phase(),leader);
      old_mbox1= false; <@\label{ln:ackend}@>
      }
\end{lstlisting}
\end{minipage}
\caption{Synchronous Paxos-like Leader election in \newmodel.  }
\label{fig:sync-leader-election}
\end{figure}

\paragraph{A computational model for canonic executions.}

Since we can reduce an asynchronous execution to a canonic one, our
     goal is to re-write an asynchronous protocol into a program with
     round-based semantics.
For this we first need to establish a programming model for canonic
     (round-based) protocols.
Several \emph{round models} exist in the
literature~\cite{DLS88:jacm,Gafni98,SW89:stacs,Charron-BostS09}.
We adapt ideas from these models for our needs, and introduce our new
     \newmodel\ model.
It allows us to express a more fine grained modeling of faults,
     network timing, and sub-routines.
The closest model from the literature is the Heard-Of
     Model~\cite{Charron-BostS09} that \newmodel\ extends to
     multi-shot algorithms (multiple inputs received during the
     executions) and has a compositional semantics based on
     synchronized distributed procedure calls.
Figure~\ref{fig:sync-leader-election} shows the \newmodel\ program
     obtained from Figure~\ref{fig:async-leader-election}.

\begin{figure}[tp]
  \vspace{-3eX}
  \centering
  \includegraphics[scale=0.30]{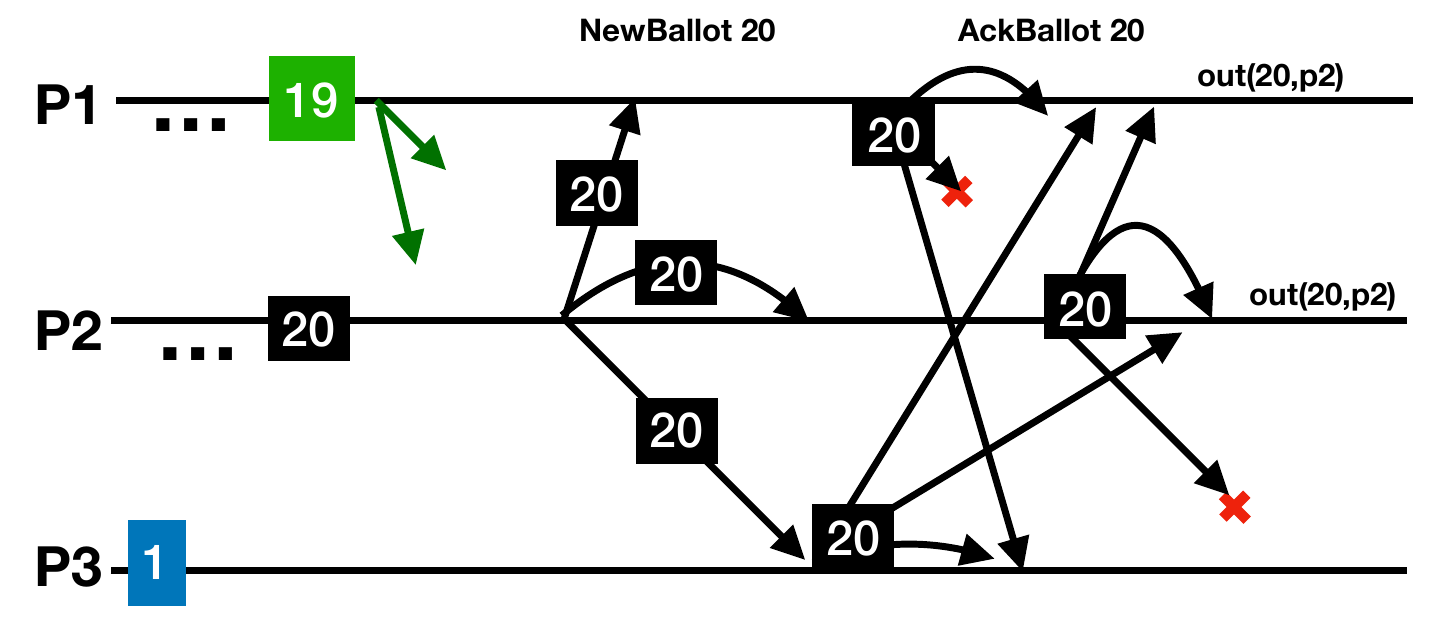}\hspace{1cm}
  \includegraphics[scale=0.30]{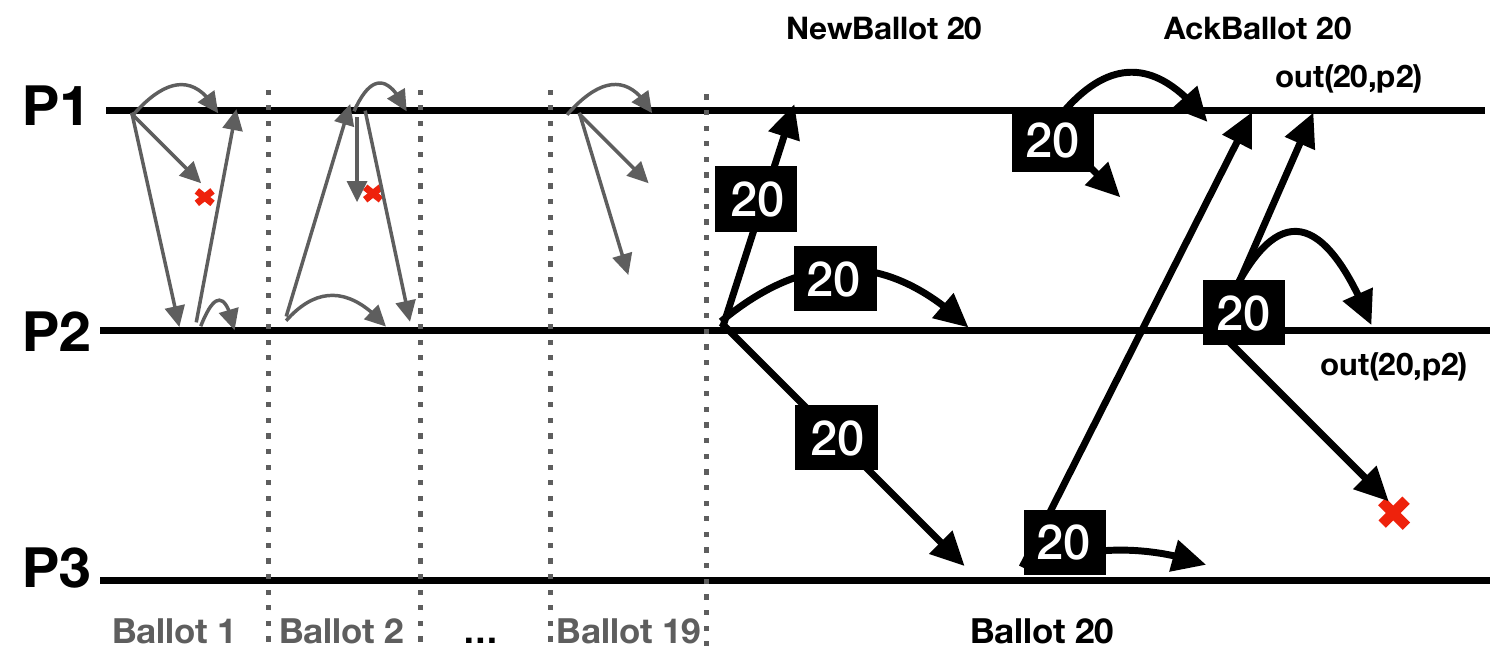}
\vspace{-1eX}
\caption{Execution with jumps}
\label{fig:asyncLE_Exec-jump}
\end{figure}


The interesting feature that abstracts away faults and timeouts are
     the so-called $\mathit{HO}$ sets.
For each round  $(b,\ell)$ and each process~$p$, the set
     $\mathit{HO}(p,(b,\ell))$ contains the set of processes from which $p$
     hears of in that round, i.e., whose messages show up in
     \texttt{mbox} in, e.g., line~\ref{ln:updatenew}.
For instance, in the simplest case, if a message from a process $q$ is
     lost, it just does not appear in the $\mathit{HO}$ set.
But also the forward jumping is accounted for.
Consider the execution on the left of
     Figure~\ref{fig:asyncLE_Exec-jump}. While the processes P1 and P2
     made progress,
P3 was disconnected from them. Then, while locally still being in
     ballot~1, P3 receives a message for ballot~20.
The execution on the right is an execution in \newmodel\ and the jump
     by Process~P3 is captured by $\mathit{HO}(P3,(i,j)) = \emptyset$, for
     $(1,\mbox{\texttt{NewBallot}})\le (i,j) <(20,\mbox{\texttt{NewBallot}})$.
For all the skipped rounds, \texttt{mbox} evaluates to the empty list.

We have augmented the Heard-Of Model with  in() and out() primitives for
     multi-shot algorithms such as state machine replication where the
     system gets commands from an external client and should output
     results.

\paragraph{Computing the canonic form.}

Having defined the round-based semantics of \newmodel,  we introduce a
     rewriting procedure.
It takes as input the asynchronous protocol together with the
     annotations that have been checked to entail a canonic form, and
     produce as output the protocol rewritten in
     \newmodel.

The main challenges for the rewriting come from the different possible
     control flows of the programs and their relation to the abstract
     notion of time.
Due to branching (e.g., in line~\ref{ln:checkleader}), code that
     appears in different places may be required to be composed into
     code for the same round; e.g., line~\ref{ln:settrue} on the
     leader branch belongs to the same round as
     line~\ref{ln:folsettrue} on the follower branch and will end up
     in line~\ref{ln:settruesync} of the canonic form. (More
     precisely, the generated code has
     branching within a round and the statement of
     line~\ref{ln:settruesync} appears in both branches. We simplified
     the code given
     in Figure~\ref{fig:sync-leader-election} for better readability.)
In addition,  there are jumps in the local time;
     cf.~Figure~\ref{fig:asyncLE_Exec-jump}.
This corresponds in the \newmodel\ to phases and rounds that are
     skipped over by a process, that is, it neither receives nor sends
     messages, and maintains (stutters) its local state.
The asynchronous statements before and after a jump must be properly
     mapped to rounds in \newmodel.
We address these issues in Section~\ref{sec:rewriting} and have
     implemented our solution.
We used it to automatically generate \newmodel\ code for several
asynchronous protocols.

\paragraph{Benefits of a round-based synchronous normal form.}

The generated \newmodel\ code represents a valuable design artifact.
First, a designer can check whether the implementation meets the
     original intuition.
For instance, the left execution in Figure~\ref{fig:asyncLE_Exec-jump}
     gives a typical asynchronous execution.
In papers on systems, designers explain their
     systems with well-formed executions like the one on the right.
The designer can check with the \newmodel\ code whether the asynchronous
     protocol implements the intended ballot and round structure, and
     whether phase jumps can occur only at intended places.
Second, it helps in comparing protocols: Different ways to implement
     the branching due to roles (e.g., leader, follower) leads to
     different asynchronous protocols.
If different asynchronous protocols have the same canonic form,  then
     they encode the same distributed algorithm.

Finally, the canonic form paves the way for automated verification.
The specification of the running example (in the asynchronous and
     synchronous version) is that in any ballot $b$, if two processes
     $p_1$ and $p_2$ find that leader election was successful
     (i.e., their $\texttt{log\_ballot}$ entry is true), then they agree
     on the leader:
\begin{multline}
\forall p_1,p_2,\ \forall b\  (\texttt{log\_ballot}[p_1][b] = true \land \texttt{log\_ballot}[p_2][b] = true) \limp \\
 \texttt{log\_leader}[p_1][b] = \texttt{log\_leader}[p_2][b] \label{propintro}
\end{multline}

To prove this property in the asynchronous model, already for the first ballot, i.e.,
$b=1$, one needs to introduce auxiliary variables to be able to state an inductive invariant.
%
A process elects a leader if it receives more than $n/2$ messages
     having the same \emph{leader id} as payload, so one needs
     to reason about the set of messages they received.
As discussed in~\cite{SergeyWT18}, this can only be achieved by introducing
     auxiliary variables that record the complete
     history of the message pool.
Then one can state an invariant over the 
     asynchronous execution that relates the local state of processes
     (decided or not) with the message pool history. 

The proof of the same property for the synchronous protocol requires
     no such invariant.
Due to communication closure, no messages need to be maintained after
     a round terminated, that is, there is no message pool.
One just needs to consider the transition relation of a phase, or ballot
     (conjunction of two rounds).
The global state after the transition, that is, at the end of the
     phase, captures exactly which  processes elected a leader in the considered phase.

In general, to prove the specification, we need invariants
     that quantify over the ballot number~$b$.
As processes decide asynchronously, the proof of ballot~1 for some
     process $p$ must refer to the first entry of \texttt{log\_ballot}
     of processes that might  already be in ballot~400.
Thus, the invariants need to capture the complete message history and
     the complete local state of processes. 
The proof in the synchronous case is modular:  For any two phases, messages do
     not interfere and processes write to different ballot/phase
     entries. Therefore the agreement proof for one ballot generalizes for all
           ballots.


Many verification techniques benefit from a reduction from
     asynchronous to synchronous.
In particular, the model checking techniques in~\cite{MaricSB17,TS11}
     are desingned specifically for the Heard-Of
     model~\cite{Charron-BostS09,BielyWCGHS07}, and can be applied on
     our output.
Theorem provers like Isabelle/HOL where successfully used to prove
     total correctness of algorithms in the Heard-Of
     model~\cite{Charron-BostDM11}.
We used deductive verification methods for the Heard-Of model~\cite{vmcai} and proved the (partial)
     correctness of the synchronous version of the running  example
     (and other protocols).




\section{Asynchronous protocols}
\label{ssec:async-sem}

\begin{figure}[t]
\scalebox{.8}{
  \begin{minipage}{0.50\linewidth}\centering
\begin{tabular}{ r c l l}
    e & := &   & expression\\
      & |  & c  & constant \\
      & |  & x  & variable \\
      & |  & f($\vv{e}$) & operation\\[2mm]
types & :=  & $\pidtype$ & process Id  \\
      & & $\type$ & user defined \\
      &&& or primitive type \\
      & &  $\paytype$ & payload type  \\
      & & p : $\pidtype$, & m : $\paytype$ \\
      & & Mbox: & set of $(\paytype, Pid)$ \\
[2mm]
   P & := & $\Pi_{p\in \mathcal{P}id} [S]_p$ & protocol \\
\end{tabular}
\end{minipage}
\begin{minipage}{0.48\linewidth}\centering
\begin{tabular}{ r c l l}
    S & :=  &   & statement \\
      & | & S ; S &  sequence  \\
      & | & x := e &  assignment  \\
      & | & reset\_timeout(e) & reset a timeout\\
      & | & send(m,p) | send(m,*) &  send message \\
      & | & (m,p) := recv() & receive message \\
      & | & if e then S else S \\
      & | & while true S\\
      & | & break \\
      & | & continue \\
      & | & x = in() & client entry \\
      & | & out(e) & client output\\
\end{tabular}
\end{minipage}}
  \caption{Syntax of asynchronous protocols.}
  \label{fig:syntax-async}
\end{figure}

Protocols are written in the core language 
     in Fig~\ref{fig:syntax-async}.
All processes execute the same sequential code, which is  
     enriched with send, receive, and timeout statements.

The communication between processes is done via typed messages.
Message payloads, denoted $\paytype$, are wrappers of primitive or
     composite type.
Wrappers are used to distinguish payload types from the types of the
     other program variables.
Send instructions take as input an object of some payload type and the
     receivers identity or $*$ corresponding to a send to all
     (broadcast).
Receive statements return an object of payload type and the identity
     of the sender, that is, one message is received at a time.
Receives are not blocking.
If no message is available, receive returns $\bot$.
We assume that each loop contains at least one send or receive
     statement.
The iterative sequential computations are done in local functions,
     i.e., \texttt{f($\vv{e}$)}.
The instructions \texttt{in()} and \texttt{out()} are used to
     communicate with an external environment (processes not running
     the protocol).

The semantics of a program \prog\ is the asynchronous parallel composition of the actions performed by all processes.
Formally, the state of a protocol $\prog$ is a tuple $\tuple{s,msg}$ where:
     $s \in [P \rightarrow \vars \cup\loc \rightarrow \datatype]$
     is a valuation of the variables in \prog where the program location is
     added to the local state and
     $msg ∈ [\paytype → (P, \texttt{T},P,ℕ)
     → ℕ]$ is the multiset of messages in transit (the network may
     lose and duplicate messages).
Given a process p ∈ P , s(p) is the local state of p, which is a valuation of p’s local variables, i.e., s(p) ∈ [Vp → D]. We use a special value~⊥ to represent the state of crashed processes. When comparing local states, ⊥ is treated as a wildcard state that matches any state.

The messages sent by a process are  added to the global pool of
     messages~$msg$, and a receive statement removes a messages from
     the pool.  The interface operations \texttt{in} and \texttt{out}
     do not modify the local state of a process.
These are the only statements that generate observable events.


An execution is an infinite sequence $s$0 $A$0 $s$1 $A$1 . . . such that $\forall i ≥ 0$,
$si$ is a protocol state, $Ai ∈ A$ is a local statement and $(si \xrightarrow{Ai} si+1)$ is a transition of the form
$\tuple{s,msg} \overset{I,O}{\longrightarrow} \tuple{s',msg'}$ corresponding to the execution of $Ai$, where $\{I,O\}$ are the observable events generated by the $Ai$ (if any).
We denote by  $\sem{\prog}\antic$ the set of executions of the protocol \prog.

%

\section{Round-based model}\label{sec:rbmodel}
We introduce \newmodel\ by first presenting the syntax and semantics
     of the intra-procedural version of \newmodel, extending it then
     to inter-procedural case.

\paragraph{Intra-procedural \newmodel\ Model.}\label{sec:intra}

\newmodel\ captures round-based distributed algorithms:
     all processes execute the same code and the computation is
     structured in rounds, where the round number is an abstract
     notion of time: processes are in the same round, and progress to
     the next round simultaneously.
We denote by $P$ the set of processes and $n=|P|$ is a parameter.
Faults and timeouts are modeled by messages not being received.
In this way the central concept is the Heard-Of set, $\mathit{HO}$-set for
     short, where $\mathit{HO}(p,r)$ contains the processes from which
     process~$p$ has \emph{heard of}\dash---has received messages
     from\dash---in round~$r$.



     \paragraph{Syntax.}

         \begin{wrapfigure}{r}{0.5\textwidth}
     \centering
     {\small
     \begin{tabular}{rcl}
     {\em protocol}      & ::= &  {\em interface} {\em variable}$^*$ {\em init} {\em phase} \\
     {\em interface}        & ::= & \texttt{in}: () $\rightarrow$ {\em type} | \texttt{out}: {\em type} $\rightarrow$ () \\
     {\em variable}         & ::= & {\em name}: {\em type} \\
     {\em init}             & ::= & \init: () $\rightarrow$ $[P → V → \mathcal{D}]$ \\
     {\em phase}            &::= & {\em round}$^+$ \\
     {\em round}$_{\sf T}$ & ::= & \send: $[P → V]$ $\rightarrow$ $[P \rightharpoonup {\sf T} ]$\\
                && \update: $[P \rightharpoonup {\sf T} ] \times [P → V]$ \\
                &&$\rightarrow$ $[P → V]$ \\
     \end{tabular}}
     \caption{\newmodel\ syntax.}
     \label{fig:HO-syntax}
     \vspace{-3eX}
     \end{wrapfigure}

A \newmodel\ protocol is composed of local variables, an
     initialization operation $\init$, and a non-empty sequence of
     rounds, called phase.
The syntax  is given in
     Fig.~\ref{fig:HO-syntax}.
A round is an object with a send and update method, and the phase is a
     fixed-size array of rounds.
     Each round is parameterized by a type \textsf{T} (denoted
     by  {\em
     round}$_{\sf T}$) which represents the payload of the messages.
The \send function has no side effects and returns the messages to be
     sent, a partial map from receivers to payloads,   based on the
     local state of each sender.
The \update function, takes as input the received messages, i.e., a
     partial map from senders to payloads, and updates the local state
     of a process.
It may communicate with an external client via
     $\mathtt{in}$, which returns an input value,
     and $\mathtt{out}$ which outputs a a value to
     the client.
For data computations,  \update uses iterative control structure only
     indirectly via auxiliary functions,  like $\texttt{all\_same}$ in
     the running example, whose definition we~omit.


\paragraph{Semantics.} The set of executions of a
\newmodel\ protocol is defined by the execution of the \send
     and \update functions of the rounds in the phase array in a loop, starting
     from the initial configuration defined by \init.

A protocol state is a tuple $\tuple{SU,s,r,msg,P,\ho}$ where:
\begin{itemize}
\item
$P$ is the set of processes executing the protocol;
\item
$SU ∈ \{Snd,Updt\}$ indicates if the next operation is send or update;
\item
$s ∈ [P → V → \mathcal{D}]$ stores the  process local states;
\item
$r ∈ ℕ$ is the round number, i.e., the counter for the executed rounds;
\item
$msg ⊆ 2^{P,{\sf T},P}$ stores the in-transit messages, where ${\sf T}$ is the type of the message payload;
\item
$\ho ∈ [P  → 2^P]$  evaluates the \ho-sets for the current round.
\end{itemize}

\begin{figure}[t]
          {\scriptsize
          \begin{mathpar}

          \inferrule[Start]{
              \overset{\init()}{\longrightarrow} s(p)
            }{
              \nothing \overset{∅,\{\init_p()\mid p∈P\}}{\longrightarrow} \tuple{Snd,s,0,∅,P,\ho}
            }\hspace{0.2cm}
          \inferrule[Send]{
            ∀ p∈P.\  s(p) \overset{\phase[r].\send(m_p)}{\longrightarrow} s(p) \\
             msg ={\begin{array}{c}\{ (p,t,q)\mid \\ p ∈ P ∧ (t,q) ∈ m_p \}\end{array}}
          }{
            \left<Snd,s,r,∅,P,\ho \right>
            \xrightarrow[p ∈ P]{\{\send_p(m_p)\},∅}
              \left< Updt,s,r,msg, P,\ho' \right>
          }


          \inferrule[Update]{
            ∀ p∈P.\ mbox_p = \{ (q,t)\mid (q,t,p) \in msg ∧ q \in \ho(p) \} \\\hspace{-2mm}
            ∀ p∈P.\ s(p) \overset{\phase[r].\update(mbox_p)), o_p}{\longrightarrow} s'(p) \\
            r' = r+1  \\
            O = \{o_p\mid p ∈ P\}
          }{
            \tuple{Updt,s,r,msg,P,\ho} \overset{\{\update_p(mbox_p)\mid p ∈ P\},O}{\longrightarrow} \tuple{Snd,s',r',∅,P,\ho}
          }
          \end{mathpar}
          }
          \vspace{-2eX}
 \caption{\newmodel\ semantics.}
 \label{fig:HO-semantics}
  \vspace{-3eX}
  \end{figure}

The semantics is shown in Figure~\ref{fig:HO-semantics}.
Initially the system state is undefined, denoted by $\nothing$.
The first transition calls the \init operation on all processes~(see
     \textsc{Start} in Fig.~\ref{fig:HO-semantics}), initializing the
     state: The round is~$0$, no messages are in the system.
\textsc{Start} brings the system into a $Snd$ state that requires the
     next transition to be a \textsc{Send}.
After that, an execution alternates \textsc{Send} and \textsc{Update}
     transitions.
In the \textsc{Send} step, all processes send messages, which are
     added to a pool of messages $msg$, without modifying the local
     states.
The values of the $HO$ sets are updated non-deterministically to be a
     subset of $P$.
The messages in $msg$ are triples of the form (sender, payload,
     recipient), where the sender and receiver are processes and the
     payload has type \texttt{T}.
The triples are obtained from the map returned by \send to which we
     add the identity of the process that executed \send.
In an \textsc{Update} step, messages are received and the \update
     operation is applied in each process.
A message is lost if the sender's identity does not belong to the $HO$
     set of the receiver.
The set of received messages is the input of \update.
If the processes communicate with an external process, then \update
     might produce observable events~$o_p$.
These events correspond to calls to $\mathtt{in}$, which returns an
     input value, and $\mathtt{out}$ that sends
     the value given as parameter to the client.
The communication with external processes is non-blocking; we assume
     that the function  $\mathtt{in}$ always returns a value when
     called.
At the end of the round, $msg$ is purged and $r$ is incremented by~1.

\begin{example}
The right diagram of Fig.~\ref{fig:asyncLE_Exec-jump1} corresponds to
     an execution of the \newmodel\ protocol
     in~Fig.~\ref{fig:sync-leader-election}.
The \textsc{Send} step of round \texttt{AckEpoch} consists of process
     P3 sending in line~\ref{ln:leadsend}, and the environment
     dropping its messages to P1 and P2.
As they do not receive messages, \textsc{Update} does not
     result in a state change due to
     line~\ref{ln:condupd1}.
Hence \texttt{old\_mbox1} does no change so that the guard in
     line~\ref{ln:condupd1} evaluates to false at P2 and P3, so that
     they do not send in the \texttt{AckEpoch} round.
\end{example}

\paragraph{Inter-procedural \newmodel\ Model.}

    We introduce \emph{distributed procedure calls} to
    capture realistic examples.
In Multi-Paxos~\cite{generalizedpaxos} processes agree on a order over
     client commands.
This order is stored in a local log, that  contains the commands
     received/committed so far.
Consider Figure~\ref{fig:compHO-exec}.
Here  a new leader gets elected with a \texttt{NewBallot} and
     \texttt{AckBallot} message exchange, almost as in our example in
     Section~\ref{sec:glance}.
The difference is the \texttt{AckBallot} round where followers (1)
     send only  to the leader instead of an all-to-all communication,
     (2) the message payload contains the current log of the follower.
Then the leader computes the longest log and sends it to its
     followers, in the third round called \texttt{NewLog}.
Those that receive the new log start a subprotocol.
The subprotocol iterates through an unbounded number of phases each
     consisting of a sequence of rounds, \texttt{Prepare},
     \texttt{PrepareOK}  and \texttt{Commit}, in which the replicas
     put commands in their logs.
Iteratively, the leader takes a new input command from the client and
     forwards it to the replicas using a \texttt{Prepare} message.
Followers reply with  \texttt{PrepareOK} acknowledging the reception
     of the new command.
If the leader receives $n/2$ acknowledgements it sends a
     \texttt{Commit} message, otherwise it considers its quorum lost
     and returns to leader election.
A follower that does not receive a message from the leader, considers
     the leader crashed, and control returns from the subprotocol to
     the leader election~protocol.

   \begin{wrapfigure}{r}{0.5\textwidth}
  \includegraphics[scale=0.50]{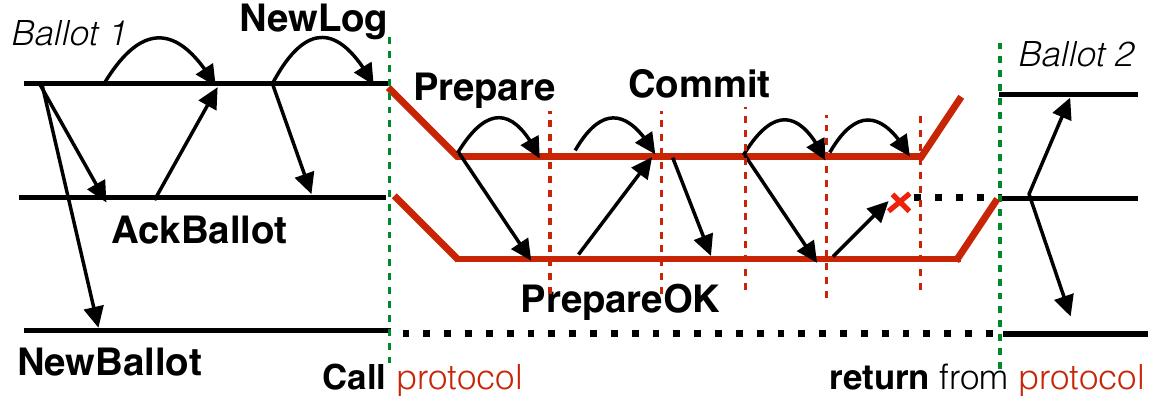}
  \caption{Inter-procedural execution}
  \label{fig:compHO-exec}
  \vspace{-3eX}
     \end{wrapfigure}

We only sketch the model.
The inter-procedural \newmodel\ protocol differs from its
     intra-procedural version only in the \update function.
A process may call another protocol and block until the call to this
     other protocol returns.
An $\update$ may call at most one protocol on each path in its control
     flow (a sequence of calls can be implemented using multiple
     rounds).
Due to branching, only a subset of the processes may make a call in a
     round.
Thus, an inter-procedural \newmodel\ protocol is a collection of
     (inter/intra-procedural) non-recursive \newmodel\ protocols, that
     call each other, with a main protocol as entry point.

%
%

\section{Formalizing Communication Closure using Tags}

We now introduce tags that use so-called synchronization variables to
     annotate protocols.
A tagging function induces a local decomposition of any execution,
     where a new block starts whenever the evaluation of the
     synchronization variables changes.
(Recall the set $T$ for our example in Section~\ref{sec:glance}.) This
     tagging thus represents a novel formalization of
     \emph{communication-closed} protocols using syntactic definitions
     of local decompositions.



\begin{definition}[Tag annotation]
  For a protocol $\prog$, a {\em tag annotation} is a tuple
  $(\svars,\tags, \tagm)$:
   \begin{itemize}
     \item $\svars=(v_1,v_2,\ldots,v_{|\svars|})$ is a tuple of fresh variables,
     \item $\tags:\loc \rightarrow [\svars \overset{\mathit{injective}}\rightharpoonup \vars]$, is a function that annotates each control location with a partially defined injective function, that maps $\svars$ over protocol variables, and
     \item $\tagm:\paytype  \rightarrow [\svars \overset{\mathit{injective}}\rightharpoonup {\sf T}]$ is an injective partially defined function, that maps variables in $\svars$ to components of the message type ${\sf T}$ (of the same type).
   \end{itemize}
  The evaluation of a tag over \prog's semantics is denoted $(\sem\tags\antic,\sem\tagm\antic)$, where
  \begin{itemize}
    \item $\sem\tags\antic:\Sigma\ra[\svars\ra \datatype \cup \bot]$,
    is a function over the set of local process states, $\Sigma=\bigcup_{s\in\sem\prog\antic}\bigcup_{p\in P}s(p)$,
    defined by
      $\sem\tags\antic_s = (d_1,\ldots, d_{|\svars|})$, with
        \begin{itemize}
          \item $d_i=\sem{\sf x_i}\antic_s$ if ${\sf x_i}= \tags(\sem\pc\antic_s)(v_i)\in\vars$, where  $v_i$ is the i$^{th}$ variable in $\svars$ and $\pc$ is the program counter,
          \item otherwise $d_i=\bot$.
        \end{itemize}
    \item $\sem\tagm\antic: \paytype\ra \dompaytype \ra[\svars\ra \datatype\cup \bot]$ is a function that for any value $m=(m_1,\ldots,m_t)$ of message type ${\sf T}$ associates a tuple
         $\sem\tagm\antic_{m:{\sf T}}=(d_1,\ldots,d_{|\svars|})$ with
      \begin{itemize}
        \item $d_i=m_j$ if $m_j = \tagm({\sf T})(v_i)$,
          where ${v_i}$ is the i$^{th}$ variable in $\svars$ and $\tagm({\sf T})$, the mapping of \svars over the message type {\sf T}, is defined in~$v_i$;
        \item $d_i = \bot$, otherwise.
      \end{itemize}
  \end{itemize}
\end{definition}

\begin{example}
For the protocol in Fig.~\ref{fig:async-leader-election}, we consider
     the tag annotation over two variables $(v_1,v_2)$ that at all
     control locations associates $v_1$ with the ballot number, and
     $v_2$ with \texttt{label}.
The first two components of messages of type $\mathtt{(int,enum,int)}$
     are mapped to $(v_1,v_2)$.
A message $m=(3,\mathtt{NewBallot},5)$ (sent in
     line~\ref{ln:sendelect}) is evaluated by $\tagm$ into
     $(3,\mathtt{NewBallot})$.
The state tag evaluates into $(3,\mathtt{NewBallot})$ if the value of
     the variable $\mathtt{ballot}$ is $3$.
\end{example}

We characterize tag annotations that imply
     communication closure:

\begin{definition}[Synchronization tag]\label{def:synctags}
Given a program $\prog$, an annotation tag $(\svars,\tags,\tagm)$
is called {\em synchronization tag} iff:
\begin{enumerate}[label=(\Roman*.)]
  \item for any local execution $\pi=s_0 A_0 s_1 A_1 \ldots \in \sem\prog\antic_{p}$ of a process $p$ in the semantics of $\prog$,
  $\sem\tags\antic_{s_0}\sem\tags\antic_{s_1}\sem\tags\antic_{s_2}\ldots$
  is a monotonically increasing sequence of tuples of values
  w.r.t. the lexicographic order.
\label{sync:sequence}

\item for any local execution $\pi\in \sem\prog\antic_{p}$, if $s \overset{send(m,p)}\longrightarrow s'$ is a transition of $\pi$, with $m$ a value of some message type, then $\sem\tags\antic_s = \sem\tagm\antic_{m}$
  and $\sem\tags\antic_s = \sem\tags\antic_{s'}$.
\label{sync:send}

\item for any local execution $\pi\in \sem\prog\antic_{p}$, if $s \xrightarrow{{(m,p)=\mathit{recv}()}} sr$ is a transition of $\pi$, with $m$ a value of some message type, then
\begin{itemize}
  \item if $(m,p)\neq\texttt{null}$ then
      \begin{itemize}
        \item $\sem\tags\antic_s\leq\sem\tagm\antic_{m}$  if
          $\tagm({\sf T})$ is surjective ({\sf T} is a message type.\\ Moreover,
          $\{\{\sem\tagm\antic_{m}\mid m\in \sem{\sf Mbox}\antic_s\}
\subseteq \{  \sem\tags\antic_s , \sem\tagm\antic_{m}   \}        $.
        \item $\sem\tags\antic_s=\sem\tagm\antic_{m}$, otherwise. Also $\sem\tags\antic_s=\sem\tags\antic_{sr}$.
      \end{itemize}
  \item if $(m,p)=\texttt{null}$ then $s=sr$.
\end{itemize}
\label{sync:receive}

\item for any local execution $\pi\in \sem\prog\antic_{p}$,
  if $s\overset{\sf stm}\rightarrow s'$ is a transition of $\pi$ such that
  \begin{itemize}
     \item $s\neq s'$, $s\mid_{\paytype,\svars} = s'\mid_{\paytype,\svars}$, that is, s and s' differ on the variables that are neither of some message type nor synchronization variables,
     \item or \texttt{stm} is a {\em send, break, continue, or out()},
   \end{itemize}
   then $\sem\tags\antic_{s}=max\{\sem\tagm\antic_{m}\mid m\in \sem{\sf Mbox}\antic_s\}$,  for all {\sf Mbox:Set({\sf T})}, {\sf T}$\in\paytype$, with $\sem{\sf Mbox}\antic_s \neq \emptyset$. That is, observable state changes and sends happen only if the state tag matches the maximal received message tag.
  \label{sync:skip}

\end{enumerate}

\label{def:synctag}

 If an annotation tag is a synchronization tag, the variables that annotated the protocol are called {\em synchronization variables}.
\end{definition}

Condition~\ref{sync:sequence} states that the variables incarnating
     the abstract time are not
     decreased by any local statement.
Condition~\ref{sync:send} states that any message sent is tagged with
     a timestamp that equals the local time of the current state.
Condition~\ref{sync:receive} states that any message received and
     stored is tagged with a timestamp greater or equal than the
     current time of the process.
All messages timestamped with greater values than the local time, must have equal timestamps.
     with the tag of the state where it is received.
Finally, \ref{sync:skip} states if messages from future rounds are
     stored in the reception variables, any statement that
     is executed in the following must not change the
     observable state, but rather increase the tag until the process
     has arrived at the maximal time it received a message from.

\paragraph{Tags and \newmodel\ protocols.}

An intra \newmodel\ protocol defined in Section~\ref{sec:rbmodel}
     executes a (infinite) sequence of phases, each consisting of a
     fixed number of rounds.
It is thus natural to annotate the code of an asynchronous protocol with a tag
     $(\mathtt{phase}, \mathtt{round})$.
In an inter \newmodel\ protocol,
     within a round, processes may call an inner \newmodel\ protocol.
Here, an instance of an inner round can  be identified by phase and
     round of the outer (calling) protocol, and phase and
     round of the inner protocol.
We are thus led in the following to consider tags that capture this
     structure:

We start with preliminary definitions.
Given two values $a\in (\mathcal{D}_A,\prec_A)$ and $b\in
     (\mathcal{D}_B,\prec_B)$, $\next(a,b) = (sa,sb)$ where (1) if $b$
     is the maximum value in $(\mathcal{D}_B,\prec_B)$, then $sb=0$
     where $0$ is the minimum value in $(\mathcal{D}_A,\prec_A)$ and
     $sa$ is the successor of $a$ w.r.t.
the order $\prec_A$, denoted $\next(a)$;
(2) else $a=sa$ and $sb=\next(b)$ is the successor of $b$ in  $ (\mathcal{D}_B,\prec_B)$.

\begin{definition}[\newmodel\ synchronization tag]
\label{def:synctagCOMPHO}
Given a protocol \prog annotated with a synchronization tag $(\svars,\tags,\tagm)$,
the tag is called
{\em \newmodel\ synchronization tag} if $\svars$ has an even number of variables, i.e., $\svars=(v_1,v_2,\ldots, v_{2m-1}, v_{2m})$, such that each pair $(v_{2i-1},v_{2i})$ has a different type (at least on one of the components) and
\begin{itemize}
  \item $v_{2i}$ takes a constant number of values, forall $i$ in $[1,m]$,
  \item the monotonic increasing order is refined; for any local execution $\pi=s_0 A_0 s_1 A_1 \ldots \in \sem\prog\antic_{p}$ of a process $p$ in the semantics of $\prog$,
  $\sem\tags\antic_{s_0} \leq_{\next} \sem\tags\antic_{s_1}$ where
  $(a_1,a_2,\ldots, a_{2m-1}, a_{2m} )\leq{\next} (a'_1,a'_2,\ldots, a'_{2m-1}, a'_{2m})$  iff
  \begin{itemize}
    \item if $(a_{2i-1}',a_{2i}') \geq succ((a_{2i-1},a_{2i}))$ then
    $(a_{2j-1}',a_{2j}') = (a_{2j-1},a_{2j})$ for all $j<i$
    and $(a_{2j-1}',a_{2j}') = (\bot,\bot)$ forall  $j>i$.
    \end{itemize}

\end{itemize}
Further, if $(a_{2i-1}',a_{2i}') = (\next(a_{2i-1}),0)$ or $(a_{2i-1}',a_{2i}') = succ((a_{2i-1},a_{2i}))$, the tag is called \emph{incremental}.

For every $1\leq i \leq m$, $v_{2i-1}$ is called a \emph{phase tag} and  $v_{2i}$ is called \emph{round tag}.
\end{definition}



\subsection{Verification of synchronization tags}
\label{ssec:verif}

Given a protocol \prog annotated with a (\newmodel-) tag
     $(\svars,\tags,\tagm)$, checking that the tag is a (\newmodel-)
     synchronization tag reduces to checking a reachability problem on
     the local code, that is, in a sequential system.

The non-sequential instructions are the sends and receives, appearing
     in Conditions~\ref{sync:send} and~\ref{sync:receive} of Definition~\ref{def:synctags}.
Checking that sent messages are tagged with the tag of the state they
     are sent in, that is Condition~\ref{sync:send}, reduces to checking equality between local
     variables: the components (tagged by \tagm) of a message type
     variable $m$ and the local variables associated by \tags at the
     control location that sends $m$.
Recall that $send$ does not modify the local state, so, it can be
     replace with an
     assert corresponding to the aforementioned equality.

Checking that messages with lower tags are dropped, that is,
     Condition~\ref{sync:receive}, is done by checking that the
     messages that are added to \texttt{mbox} have values (on the
     tagged components) greater than or equal to the local variables
     associated by $\tags$ at the control location where the addition
     occurs. 
We assume that $recv$ may return any message and we check that the filters 
that guard the message's addition to the mailbox respect the order relation w.r.t. the state tags. 
This is again expressed by a state property that relates message fields with tag variables.

Conditions~\ref{sync:sequence} and~\ref{sync:skip} in
     Def.~\ref{def:synctag} (and Def.~\ref{def:synctagCOMPHO}) translate into transition invariants over the synchronization
     variables. They state that the lexicographic order (monotonic or
     increasing) is preserved by any two consecutive assignments to
     the synchronization variables.

We automated these checks with the static verifier Verifast, and
     report in Section~\ref{sec:example} on our experiments.



\section{Reducing an asynchr.\ execution to its canonic form}
\label{ssec:reduction}

After having introduced synchronization tags, we now show that any
     execution of an asynchronous protocol that has a synchronization
     tag can be reduced to a canonic execution.
The proof proceeds in several steps, where in each step we will obtain
     a more restricted execution.
The steps are as follows:
\begin{description}

\item[Asynchronous executions.] We start with an asynchronous
     execution $\asyncexec{}\in \sem\prog\antic$ as defined in
     Section~\ref{ssec:async-sem}.
Due to asynchronous interleavings, an action at process $p$ that
     belongs to round~$k$ may come before an action at some other
     process $q$ in round~$k'$, for~$k' < k$.

\item[Big receive.] In order to capture jumping forward in
     rounds, we will regroup statements at different process to arrive
     at an asynchronous execution, where for each process a sequence
     of receive statements (followed by local computations ${\sf stm}$
     for a jump) appears in a block.
Thus, we can replace these blocks by a single atomic $Receive$.
The resulting executions we denote by~$\sem\prog_\mathsf{Rcv}\antic$.

\item[Monotonic executions.]
We reduce asynchronous executions with
Big receive semantics to execution where all tags are (non-strictly)
monotonically increasing. As a result, all actions for round~$k'$
appear before all actions for all rounds $k$, for~$k' < k$.

\item[Round-based executions.] We reduce
     monotonic executions to \newmodel\ executions as defined in
     Section~\ref{sec:rbmodel}.

\end{description}

In each step, we maintain the following important property between the
original execution and the execution we reduce to:

\begin{definition}[Indistinguishability]
Given two executions $\aexec$ and $\aexec'$ of a protocol \prog, we
     say a process $p$ cannot distinguish locally between $\aexec$ and
     $\aexec'$ w.r.t.\ a set of variables $W$, denoted
     $\aexec\simeq_p^W\aexec'$, if the projection of both executions
     on the sequence of states of $p$, restricted to the variables in
     $W$, agree up to finite stuttering, denoted, $\aexec\prj{p,W}
     \equiv \aexec'\prj{p,W}$.

Two executions $\aexec$ and $\aexec'$ are \emph{indistinguishable}
     w.r.t.\ a set of variables $W$, denoted $\aexec\simeq^W\aexec'$,
     iff no process can distinguish between them, i.e., $\forall
     p.\ \aexec\simeq_p^W\aexec'$.
\end{definition}

We focus on indistinguishability because it preserves  so-called local
     properties~\cite{Chaouch-SaadCM09}, or equivalently properties
     that are closed under local stuttering.
Important fault-tolerant distributed safety and liveness
     specifications fall into this class: consensus, state machine
     replication, primary back-up, k-set consensus,~etc.

\begin{definition}[Local properties]\label{def:local}
A property $\phi$ is \emph{local} if for any
two  executions $a$ and $b$ that are  indistuingishable
 $a \models \phi$ iff $b \models \phi$.
\end{definition}

In the following we will denote by $S_i$
     a global state and by $s_i(p)$ the local state of process $p$ in
     the global state $S_i$.

\paragraph{Reducing Asynchrony to Big receive.} \label{sec:bigreceive}

This reduction considers the receive statements.
If the local execution is of the form $\pi= \ldots s^p_i, receive_i,
     s^p_{i+1},$ $receive_{i+1}$,  $s^p_{i+1}, \ldots$, in the
     asynchronous execution, the two receive actions can be
     interleaved by actions of other processes.
Following the theory by Lipton~\cite{Lipton75}, all receive statements
     are right movers with respect to all other operations of other
     processes, as the corresponding send must always be to the left
     of the receive.
In this way, we reduce an asynchronous execution to one where local
     sequences of receives appear as block.
By the same argumentation, this block can be moved right until the
     first ${\sf stm}$ action of this process.
Again the resulting block can be moved to the right w.r.t.\ actions at
     other processes.
By repeating this argument, we get an asynchronous executions with
     blocks that consist of several receives (possibly just one
     receive) and ${\sf stm}$ statements such that at the end the
     local state tag matches the maximal received message tag, i.e., the process has jumped forward to a round from which it
     received a message.
We will subsume such a block by an (atomic) action $Receive$,
     and denote by $\sem\prog_\mathsf{Rcv}\antic$ the asynchronous
     semantics with the atomic Receive.

\paragraph{Reducing Big receive to monotonic.}\label{sec:monotonic}

\begin{theorem}
  Given a program $\prog$ if there is a synchronization
  tag $(\tags,\tagm)$ for $\prog$, then
  $\forall \asyncexec{}\in \sem\prog_\mathsf{Rcv}\antic$, if $\asyncexec{} = \ldots
  S_{i-1}, A_i^p,
  S_i, A_{i+i}^q, S_{i+i} \ldots$, and $\sem\tags\antic_{s_i(p)} >
  \sem\tags\antic_{s_{i+1}(q)}$, then $${\asyncexec{}}' = \ldots
  S_{i-1}, A_{i+i}^q,
  S'_i, A_i^p, S_{i+i} \ldots\in \sem\prog\antic.$$ Further ${\asyncexec{}}', {\asyncexec{}}$ are indistinguishable w.r.t. all protocol variables, i.e., ${\asyncexec{}}'\simeq{\asyncexec{}}'$.
\end{theorem}

\begin{proof}
From $\sem\tags\antic_{s_i(p)} >
  \sem\tags\antic_{s_{i+1}(q)}$ and \ref{sync:sequence}
follows that $p\ne q$, so that swapping
     cannot violate the local control flow.
As $p\ne q$, if $A_{i+i}^q$ is a send or a stm, the action at $p$ has
     no influence on the applicability of $A_{i+i}^q$ to $S_i$.
The only remaining case is that $A_{i+i}^q$ is a $Receive$. Only if $A_i^p$
     sends a message $m$ that is received in $A_{i+i}^q$, $A_{i+i}^q$
     cannot be moved to the left. We prove by contradiction that this is
     not the case:
By \ref{sync:send},  $\sem\tags\antic_{s_i(p)}= \sem\tagm\antic_{m}$.
By \ref{sync:receive} and \ref{sync:skip}, and the atomicity of
     Receive, $\sem\tags\antic_{s_{i+1}(q)} = \sem\tagm\antic_{m}$.
Thus, $\sem\tags\antic_{s_i(p)} = \sem\tags\antic_{s_{i+1}(q)}$ which
     provides the required contradiction to the assumption of the
     lemma   $\sem\tags\antic_{s_i(p)} >
     \sem\tags\antic_{s_{i+1}(q)}$.

The statement on indistinguishability follows
from the reduction.
\qed\end{proof}

By  inductive application of the theorem, we obtain:

\begin{corollary}
 Given a program $\prog$ if there is a synchronization
  tag $(\tags,\tagm)$ for $\prog$, then
  $\forall \asyncexec{}\in \sem\prog_\mathsf{Rcv}\antic$, there is a
monotonic
asynchronous
     execution $\seqcomp{\asyncexec{}}  = \ldots S_{i-1},$ $A_i^p,$ $S_i, A_{i+i}^q,
     S_{i+i} \ldots$, where for each $i$ and any two processes $p$ and
     $q$, $\sem\tags\antic_{s_i(p)} \le
     \sem\tags\antic_{s_{i+1}(q)}$.
\end{corollary}

The monotonic execution $\seqcomp{\asyncexec{}}$ is thus a sequential
     composition of actions of rounds in increasing order, that is, all
     actions of round~$k$ occur before all actions in round~$k+1$, for
     all~$k$.
Thus, the global state between the last round~$k$ action and the first
     round~$k+1$ action constitutes the boundary between these rounds.
In the following section we will show that we can simplify the
     reasoning within a round.

\paragraph{Reducing a Round to a Synchronous round.}\label{sec:redHO}

In order to reduce monotonic executions into $\mathit{HO}$ semantics we re-use
     arguments by~\citet{Chaouch-SaadCM09}, which we have to extend
     for asynchronous programs.
We consider distributed programs of a specific form: The local code
     within each round is structured in that first there are send, then
     Receive, and then other statements.
Similarly, we only consider protocols where it is sufficient to check
     states only when the tags change.
If the local code within a round is ``subsumed'' to a single
     local transition, we do not lose any observable events.
Rather, the subsumption is locally stutter equivalent to the original
     asynchronous semantics.

As we start from monotonic executions here, we can restrict ourselves
     to swapping actions within a round and only have to care about
     moving send and receive actions.
For this, we can use the arguments from  \cite{Chaouch-SaadCM09}: the
     send actions are left movers with respect to all other
     operations, Receive actions are left movers with all statements
     except sends.
By repeated application of their arguments, we arrive at executions
     where within a round all send actions come before all Receive
     actions, which come before all other actions.
We call these executions send-receive-compute executions:

\begin{proposition}
For each monotonic asynchronous execution, there exists an
     indistinguishable asynchronous send-receive-compute execution.
\end{proposition}

All sends are non-interfering and can thus be ``accelerated'' or
     ``subsumed'' in one global send action.
As in \newmodel\ all messages sent in a round must be of identical
     payload type, the type to be sent in the subsumed action is the
     union of the payload types.
Similar for receive.
Here, the $\mathit{HO}$ sets are defined by the processes of which a process
     received messages in its receive operations.
If in the original execution~$\aexec$ process~$p$ jumped over
     round~$r$, there are no send, receive, and local computation
     actions for $p$ in $r$.
As we require in the $\mathit{HO}$ semantics that every process performs these
     steps in each round, we have to complete the execution with
     \texttt{nop} steps for the missing rounds.
As they do not change which messages are received in the asynchronous
     execution, and which local states the processes go through, we
     again remain stutter equivalent, and obtain.

\begin{proposition}
For each asynchronous send-receive-compute execution, there exists an
indistinguishable \newmodel\ execution $\syncexec{}$, where the
     messages received in a Receive statement correspond to the $\mathit{HO}$
     sets.
\end{proposition}

Following Definition~\ref{def:local}, local properties are those closed under
indistinguishability, so that we obtain the following theorem.

 \begin{theorem}\label{thm:ho}
 If there exists a synchronization tag
       $(\svars,\tags,\tagm)$ for $\prog$, then $\forall \asyncexec{}\in
       \sem\prog\antic$ there exists an \newmodel-execution
       $\syncexec{}$ that satisfies the same local properties.
  \end{theorem}

%

%
 
\section{Code to code rewriting of Asynchronous to \newmodel}
\label{sec:rewriting}

We introduce a rewriting algorithm $\algorw$ that takes as input an
     asynchronous protocol \prog annotated with a synchronization tag
     and either produces a (inter-procedural) \newmodel\ protocol,
     denoted \newmodel(\prog), whose executions are indistinguishable
     from the executions of \prog, or aborts.

\paragraph{Replacing reception loops with atomic mailbox definition.}
\label{ssec:recloops}

We consider asynchronous protocols where message reception is
     implemented in a distinguished loop, that we refer to as
     ``reception loop''.
A reception loop is a simple \texttt{while(true)} loop, that (1)
     contains \texttt{recv} statements, (2) writes only to variables
     of message type, or containers of message type objects, (3) the
     loop is exited either because of a timeout or because some
     condition  over the message type  variables wrote in the loop holds.
The algorithm in Fig.~\ref{fig:async-leader-election} has four
     reception loops, at lines~\ref{ln:startrcv}, \ref{ln:rcvleader},
     \ref{ln:folstartrcv}, and~\ref{ln:rcvfollower}.
The exit of a reception loop is typically cardinality constraints or
     timeout, e.g., \texttt{mbox$\rightarrow$size > n/2} or
     \texttt{mbox$\rightarrow$size == 1}.

A reception loop is replaced by havoc assignments of the message
     type/container of message type variables written by the loop.
The code following the  loop is left unchanged except in the following
     cases: (1) the boolean conditions that refer to a loop timeout
     are replaced by the negation of all the other conditions to
     exit the loop; (2) if the loop does not have a timeout exit, that
     is, processes wait until all required messages are received, the
     code following the loop is wrapped into an if statement, allowing
     its execution only if the loops exit condition holds.
In the rest of the section we consider only protocols whose reception
     loops have been replaced by havoc statements.

\begin{figure}[t]
\begin{minipage}{0.49\textwidth}
\begin{lstlisting}[frame=none, language=C]
if (coord() == me){
   ballot++;
   msg *m = create_msg(ballot,label,myid); 
   if (mbox!=0 && mbox->size ==1) {
  	 ballot = mbox->message->ballot; 
	 leader = mbox->message->sender;}}
	\end{lstlisting}
	\end{minipage}
\begin{minipage}{0.49\textwidth}
  \ContinueLineNumber
	\begin{lstlisting}[frame=none, language=C]
if !(coord() == me){
 ballot++; 
 if (mbox!=0 && mbox->size ==1) {
     ballot = mbox->message->ballot; 
     leader = mbox->message->sender;}}
\end{lstlisting}
\end{minipage}
\caption{Two blocks defining round \texttt{NewBallot} in the  protocol
  from Fig.~\ref{fig:async-leader-election}.}
\label{fig:async-leader-election-blocks}

\vspace{-3eX}
\end{figure}

\paragraph{Rewriting  protocols with incremental
  synchronization tags.}
\label{ssec:rewrite1}

Let \prog be a protocol consisting only of one loop annotated with an
     incremental \newmodel\ synchronization tag $(\ph,\rd)$.
The rewriting in this section builds a (intra-procedural)
     \newmodel\ protocol in two steps: (1)  each iteration of \prog's
     loop defines a phase and (2) the code of each phase (the
     loops body) is decomposed into rounds.

Phases are matched with loop iterations if 
$\ph$ (representing the phase
     number) is increased exactly once in each iteration to its
     successive value (like the loop counter).
     To this we assume the protocol verified for strengthened annotations regarding tags monotonicity:   
the relation $\ph = old\_\ph +1$ is an invariant of the loop, where
     $old\_\ph$ is previous value of $\ph$.
If $\ph$ has initial value 1, then the phase number matches the iteration number.
Otherwise $\ph$ is shifted by a bounded value.
The communication closure induced by the tags (see Theorem~\ref{thm:ho}) ensures that two processes communicate only when they are in the same iteration.
Hence, it is sound to construct a phase $i$ by composing the $i$th loop iteration of all processes.
Within a phase it remains to locate the round boundaries.

A \newmodel\ synchronization tag, ensures that the round variable
     \rd\ takes a bounded number of values: in the running examples
     these values are \texttt{NewBallot} and \texttt{AckBallot}.
Round bounderies are defined by the beginning/end of a loop iteration
     and  the assignments to the round variable \rd.

Processes can have different behaviors in the same round, depending on
     their local state and the messages received, although they
     execute the same code and go through the same sequence of rounds.
For example, in the round \texttt{NewBallot} only the processes
     designated coordinators by the oracle send a message.
Similarly, in the \texttt{AckBallot} round only the processes that
     received a message in this round are going to update their logs.
As usual these different behaviors are captured by branching
     instructions in the loops's body, and each path in the loop's
     body identifies a possible process behavior in sequence of
     rounds.

For each value $\ell$ of \rd, to compute the code of round $\ell$, we consider each path $\pi$ in the control flow graph of the loop's body and we identify
(1)  a block of instructions (possibly empty) $B_{\ell}^\pi$: a sequence of instruction in $\pi$ that starts with $\rd = \ell$ and ends with the  instructions preceding the next assignment to \rd;
(2)  the context under which each block $B_\ell^\pi$ is executed, that is a condition $cond_\ell^\pi$ that is the conjunction of all the branches leading to $rd=\ell$ on the path $\pi$.
The $B_\ell$ is the sequential composition of all \texttt{if ($cond_\ell^\pi$) $B_\ell^\pi$} with $\pi$ path in the control flow.
Fig.~\ref{fig:async-leader-election-blocks} shows the two blocks defining round NewBallot, corresponding to the leader follower paths in the control flow graph. 

To maintain the context in which a sequence of instructions is executed, i.e., $cond^\path_\ell$, we introduce auxiliary variables.
For each variable $\texttt{x}$ in a conditional we introduce an auxiliary variable $old\_x$ (of the same type with $x$), that is assigned only once to $\texttt{x}$, i.e, $old\_x := \texttt{x}$, before the condition is evaluated.
The conditionals $cond^\path_\ell$ are defined over these auxiliary variables.
If the variable is not read without being first assigned to a default value, we can abstract  $old\_x$ to boolean. This is the case of all our benchmarks, where auxiliary variables remember values of the mailbox in previous rounds.

Moreover, if the values of the round variables do not take all values in their domain, 
each condition $cond_\ell^\pi$ is
     conjuncted with the check whether the round number of the
     \newmodel\ equals the round tag variable. 
     Intuitively, if the check fails, the asynchronous code has set the
     round tag to a future round (of the same phase), which results in skipping the
     \newmodel\ round.

Finally, the code of every round, that is, $B_\ell^\pi$ is split into a
     $Send_\ell^\pi$ block, consisting of all send statements $B_\ell^\pi$ guarded by the conditionals preceding them and 
     an $Update_\ell^\pi$ block that contains the rest of the code in $B_\ell^\pi$ except the mbox's havoc. 
     We assume $Update_\ell^\pi$ contains no send, no recv, and no assignments to message type variables. 
     Otherwise the rewriting aborts. (One could try compiler optimization techniques to reorder instructions towards the imposed order.)  

The rewriting eliminates the phase and round tag variables (if no rounds are skipped) from the
     local process variables all program locations are tagged with the
     same variables.
Reads of these variables are replaced with reads of the round,
     respectively phase number, of the \newmodel\ protocol.

\begin{example}
   For example the asynchronous protocol in Fig.~\ref{fig:async-leader-election} in Sec.~\ref{sec:glance} is rewritten into an intra-procedural \newmodel\ one, given in Fig.~\ref{fig:sync-leader-election} in Sec.~\ref{sec:glance}. However Fig.~\ref{fig:sync-leader-election} contains a simplification w.r.t. what is automatically generated. The code between the lines~\ref{ln:ackstart} to~\ref{ln:ackend} appears twice, ones if the process is leader and ones if it is not. Similar for the first round. 
\end{example}

An asynchronous protocol is structured, if receptions loops are emphasized as defined in Sec.~\ref{ssec:recloops} and the blocks associated with a round are a sequence of send followed by update statements. 

\begin{theorem}
Given a structured asynchronous protocol $\prog^{async}$ consisting of only one loop, that is
 annotated with a strictly incremental \newmodel\ synchronization tag  $(\svars,\tags,\tagm)$ of size two,
       $\svars=\{ \mathtt{phase\_tag}, \mathtt{phase\_tag}\}$,
\algorw builds an intra-procedural \newmodel\ protocol $\prog^{\newmodel}$  whose executions are indistinguishable from the executions of $\prog^{async}$.
The resulting protocol has only one phase that consists of as many rounds as the domain of evaluation of the round\_tag.
$\prog^{\newmodel}$ sends exactly the same messages as $\prog^{async}$.
\label{th:rewrite1}
\end{theorem}

\paragraph{Jumping over phases.}
\label{ssec:rewrite4}

The catch up mechanism allows processes to receive messages from
     future rounds, which leads to a jump to the received phase
     number.
Moreover, in general non-incremental tags allow processes to skip to
     future tags (which may happen, e.g., if the leader of the current
     phase is suspected to have crashed).
In this section, we reduce the problem of rewriting a protocol
     with non-incremental tags to the rewriting a protocol
     with incremental tags.

In Sec.~\ref{ssec:rewrite1} the loop counter coincides (modulo an
     initial shift) with the phase tag.
Jumping over phases potentially increases the phase tag by more than one,
     ``desynchronizing'' it from the loop counter.
To apply the rewriting from Sec.~\ref{ssec:rewrite3} we introduce
     empty loop iterations, when the loop counter is smaller than the
     phase tag, and we reinterpret the initial increasing tags over
     the new loop counter, resulting into an incremental tag
     annotation.


First we identify the jumping control locations.
These are locations where the phase tag (1) is assigned a value that
     depends on the mailbox and (2) the communication closure checks
     show that a message tag in the mailbox may be strictly greater than
     local tag; cf.~Section~\ref{ssec:verif}.
In this case the tool partitions the path with the jumping instruction into the three
sequences of instructions
\texttt{Before\_Jump}, \texttt{Jump}, and \texttt{After\_Jump}. 
Since jumps are conditional, we have to capture the cases without jump,
     where \texttt{Before\_Jump} and \texttt{After\_Jump} are both part of
     a round and the case with jump where \texttt{Before\_Jump} and
     \texttt{After\_Jump} are parts of code for different rounds.
The rewriting encodes this cases with an auxiliary boolean variable
     that non-deterministically flags a jump, and a \texttt{continue}
     statement before the jumping instruction.  

In all examples we explored, \texttt{Before\_Jump} is either empty or
     consists only of one \texttt{send} instruction.
Both cases are simpler and correspond to no code being execution in
     \newmodel\ semantics, and it is naturally captured by empty HO
     sets there (in case there are send instruction the 
     messages can always be lost).





\paragraph{Protocols with nested loops.}
     \label{ssec:rewrite3}

Let us consider a protocol \prog without reception loops.
The rewriting algorithm proceeds bottom-up: it starts rewriting the
     most inner loop using the procedure above.
For each outer loop it first replaces the nested loop with a call to
     the computed \newmodel\ protocol, and then applies the same
     rewriting procedure.
Since we considered passing by value procedure calls in the \newmodel\
     semantics, all local variables are input parameters.

Inner loops appearing on different branches may belong to
     the same sub-protocol; in other words these different loops
     exchange messages.
If  $\tags$ associates different synchronization variables to
     different loops then the rewriting builds one (sub-)protocol for
     each loop.
Otherwise, the rewriting merges the loops tagged with the same
     synchronization variables into one \newmodel\ protocol.


To soundly merge several loops into the same \newmodel\ protocol, the
     rewrite algorithm  identifies the context in which the inner loop
     is executed.

     \begin{theorem}
       Given a structured asynchronous program $\prog^{async}$  with a \newmodel\  synchronization tagging function $(\svars,\tags,\tagm)$, then $\algorw$ applied $\prog^{async}$ returns an inter-procedural \newmodel\ protocol $\prog^\mathit{\newmodel}$ whose executions are indistinguishable from the executions of $\prog^{async}$.
     \end{theorem}

\section{Experimental results}
\label{sec:example}
\newcommand{\tool}{\texttt{Translate\_to\_Sync}}

We implemented the rewriting procedure in a prototype tool and applied
the tool to several fault-tolerant distributed protocols.
The tool is available
\href{https://github.com/alexandrumc/async-to-sync-translation}
     {online}.\footnote{
       \url{https://github.com/alexandrumc/async-to-sync-translation}}
Fig.~\ref{fig:examples} summarizes our experimental results.

\paragraph{Verification of synchronization tags.}

\begin{wrapfigure}{r}{0.49\textwidth}
\vspace{-4eX}
\begin{lstlisting}[frame=none, language=C]
predicate tag_leq(int old_ballot, int old_label, int ballot, int label) = (old_ballot < ballot) || (old_ballot ==ballot && old_label<=label) ;
 predicate mbox_tag_eq(int ballot, int round, struct List* n;) =
 n == 0 ? true :
 n->message |-> ?msg &*& msg->ballot |-> ?v &*& msg->label |-> ?r  &*& msg->sender |-> _  &*&
 malloc_block_Msg(msg) &*& malloc_block_List(n) &*& n->next |-> ?next &*& n->size |-> ?s &*&
 n!=next &*&  ballot== v &*&  round== r &*& mbox_tag_eq(ballot,round, next) ;
\end{lstlisting}
\caption{Predicates in separation logic expressing the order relation over tags and the condition that all messages in the mailbox are timestamped with the receiver's local time.}
\vspace{-3eX}
\label{fig:preds}
\end{wrapfigure}

The tool takes protocols in a C embedding of the language
     from Sec.~\ref{ssec:async-sem} as input.
We use a C embedding to be able to use Verifast~\cite{verifast} for
checking the conditions in Sec.~\ref{ssec:verif}, i.e.,
the communication closure of an asynchronous protocol.
Verifast is a deductive verification tool based on separation logic for
     sequential programs.
The C embedding uses  the prototype of the functions \texttt{send}
     and \texttt{receive} (we assume their semantics is the one in
     Sec.~\ref{ssec:async-sem}).


The user specifies in a configuration file the synchronization tag by
     (i) defining the number of (nested) protocols in the input file, (ii)
     for each protocol, the phase and round variables,  and (iii) for each
     messages type the fields that encode the timestamp, i.e., the phase
     and round number.
     Fig.~\ref{fig:examples} gives the names of phase and round
     variables of published protocols we use as benchmarks.

The tool expects the input file to be annotated with assert statements
for checking the conditions in Definition~\ref{def:synctags} w.r.t.\
the tags given in the configuration file and the auxiliary annotations
     Verifast needs to prove these asserts (inductive invariants).
The annotations are defined over program variables and auxiliary
     history variables.
Auxiliary variables are necessary to encode the monotonicity of the
     tags.
The tool calls Verifast and checks that the input contains \texttt{assert
     tag\_leq(old$\_$ballot, old$\_$label, ballot, label)} after each
     (pair of) assignment(s) of the phase and round variables.
Observe that conditions \ref{sync:sequence}--\ref{sync:receive} in Definition~\ref{def:synctags}
      are numeric constraints over the phase and
     round variables, and several other tools might verify them.
However, condition \ref{sync:skip} requires reasoning about the content of the
     mailbox, a potentially unbounded data structure.
Here is where we used the strength of Verifast, to reason about dynamically allocated data
     structures.
The size of the program annotated with the proofs for the asserts is given in LoC in the column ``Annot.'' in Fig.~\ref{fig:examples}.
If all the checks are passed, then the rewriting proceeds, otherwise
the tool outputs a warning.


\begin{figure}[t]
  {\small
\begin{center}
  \begin{tabular}{l|| l | c | c | c}
    \hline
    Protocol & Tags & Annot.  & Async & Sync\\
      \hline\hline
        {\small \begin{tabular}{l}
                  Consensus\\
                 \cite[Fig.6]{ChandraT96} \\
        \end{tabular}}&
               {\small \begin{tabular}{l}
                   $\ph$ = ${r}_p$ \\
                   $\rd$= \{{Phase1, Phase2}, {Phase3, Phase4}\}\\
                  \end{tabular}} &
        661 &
        332 & 251
        \\ \hline
        Two phase commit &
        {\small \begin{tabular}{l}
                 $\ph$ = {i},  \\
                 $\rd$= \{{Query, Vote,}
                  {Commit, Ack}\}\\
                  \end{tabular}} &
        588 &
        252 &
        242 \\ \hline
    {\small \begin{tabular}{l}
                 Figure~\ref{fig:async-leader-election}$^{*,V}$ \\
       \end{tabular}} &
    {\small \begin{tabular}{l}
                 $\ph$ = {ballot},\\
                 $\rd$ = \{NewBallot, AckBallot\} \\
                 \end{tabular}} &
       650 & 255 & 110
       \\ \hline
        ViewChange$^*$~\cite{OkiL88} &
        {\small \begin{tabular}{l}
                 $\ph1$ = $\texttt{view},$ \\
                 $\rd1$ = \{StartViewChange,  \\
                \quad\quad\quad\quad DoViewChange,
                 StartView\}\\
        \end{tabular}} & 720 &
         352 &
        172 \\ \hline
        \begin{tabular}{l}
                 Normal-Op$^{V}$~\cite{OkiL88}\\
        \end{tabular}
         &
       \begin{tabular}{l}
                 $\ph$ = op\_number\\
                 $\rd$ = \{Prepare,
                  PrepareOK, Commit\}
                  \end{tabular} &
         628  &
        266&
       182 \\ \hline
        Multi-Paxos$^{*,V}$~\cite{generalizedpaxos} &
       \begin{tabular}{l}
                 $\ph1$ = ballot,\\
                 $\rd1$ = \{NewBallot, AckBallot,  NewLog\}\\
                 $\ph2$ =  op\_number,\\
                  $\rd2$ = \{Prepare,
                  PrepareOK, Commit\}
                  \end{tabular} &
        1646  &
        621 & 405
         \\ \hline
    \hline
  \end{tabular}
\end{center}
}
\caption{Communication-closed asynchronous protocols. The superscript * identifies protocols that jump over phases. The superscript V marks the protocols whose synchronous counter-part we verified. }
\label{fig:examples}
\end{figure}

\paragraph{Rewriting.} While checking the
verification tags can be done for any annotated asynchronous protocol,
     the rewriting tool checks whether the asynchronous protocol is in
     a specific form and only then translates it into \newmodel.
While in theory this is a restriction, the benchmarks in
     Fig.~\ref{fig:examples} show that well-known algorithms are rewritten
     by our tool.
For instance, the algorithm~\cite[Fig.~6]{ChandraT96} solves consensus
     using an eventually strong failure detector.
The algorithm jumps over rounds in a specific way.
If a special decision message is received, a process jumps forward to a
     decision round and outputs the decision value.
The resulting algorithm is much like Last Voting in
     \cite[Fig.~5]{Charron-BostS09}.
ViewChange is a leader election algorithm similar to the one in
     ViewStamped: unlike in the running example (unlike Paxos), in
     ViewChange processes first agree to change the current leaders,
     and than on a leader.
The phase number is the view number (like in Paxos), that is, two
     processes either agree on the identity of a leader in a view or
     they know of no leader.
     Normal-Op is the sub-protocol used in ViewStamped to implement the
     broadcasting of new commands by a stable leader.
Multi-Paxos is described in Sec~\ref{sec:rbmodel}. It is Paxos
     from~\cite{generalizedpaxos} over sequences, without fast paths,
     where the classic path is repeated as long as the leader is stable.
     In Paxos parlance, the tags for  leader election (outer protocol) are
     {\em Phase1a}, {\em Phase1b}, {\em Phase1aStart} (in this order).
     The rounds of the sub-protocol are called  {\em Phase2aClassic}, {\em Phase2bClassic}, {\em learn}.
     We considers that acceptors and leaders play also the role of learners.

Our tool has rewritten the protocols from Fig.~\ref{fig:examples}.
The implementation uses pycparser~\cite{pycparser}, a parser for the C
     language written in pure Python, to obtain the abstract syntax
     tree of the input protocol.
The last two columns of Fig.~\ref{fig:examples}  give the size in LoC
     of the asynchronous protocol without annotations and the size of
     its synchronous counterpart computed by the rewriting procedure
     from Sec.~\ref{sec:rewriting}.


\paragraph{Verification.}
We have verified the safety specification (agreement) of the
     \newmodel\ counter-parts of the running example
     (Figure~\ref{fig:async-leader-election}), Normal-Op,  and
     Multi-Paxos, by deductive verification using the Consensus Logic
     (CL for short) defined in~\cite{vmcai}.
To this, we encoded the specification and the transition relation in
     CL, and used CL's semi-decision procedure for
     satisfiability~\cite{psynctool} to discard the verification
     conditions.
For Multi-Paxos we did a modular proof.
First we prove the correctness of the sub-protocols (executed in case
     of a stable leader).
Its specification is that the logs of all processes that execute the
     sub-protocol are equal at the beginning and at the end of each
     phase (after an iteration of \texttt{Prepare},
     \texttt{PrepareOk}, \texttt{Commit}),  knowing that processes
     start the sub-protocol with equal logs.
Moreover, the sub-protocol preserves the invariant property that a
     majority of processes have the same prefix, consisting of  all
     the committed commands.
Then we prove the leader election outer loop correct.
Its specification states that there is at most one leader in a ballot
     (like in~(\ref{propintro})) and that a majority of processes have
     the same prefix, consisting of  all the committed commands.
The leader picks the longest log of its followers.
The fact that all committed values are logged by a majority of
     processes ensures that the new log proposed by the leader will
     not have lost any committed commands.
However, there are no guarantees for the uncommitted commands.

\section{Related work}
\label{sec:relred}

\newcommand{\abcd}{\cite{DBLP:journals/siamcomp/MosesR02,ChouG88,EngelhardtM05a}}

\newcommand{\roundpaps}{\cite{DLS88:jacm,Lyn96:book,Charron-BostS09,Gafni98,SW89:stacs}}

\newcommand{\syncverif}{\cite{DBLP:conf/srds/TsuchiyaS07,DBLP:conf/wdag/TsuchiyaS08,Charron-BostDM11,DBLP:journals/afp/DebratM12,vmcai,GleissenthallBR16,MaricSB17,AminofRSWZ18}}


  Our goal is to link
     synchronous or round-based models to asynchronous models via the
     notion of communication closure~\cite{EF82}.
Exploiting this for better design and simpler paper-and-pencil proofs
     was considered, e.g., in~\abcd.
Several round-based computational models are based on this
     idea~\roundpaps.
Commonly the underlying idea is to design an algorithm for
     the round-based setting and deduce results for the asynchronous.
A method that takes round-based code as input and generates
     asynchronous code was given in~\cite{DragoiHZ16}.
However, for efficiency reasons, designers often prefer to work with
     asynchronous code.
Therefore, in this paper we start from asynchronous protocols, and
     compute the round-based canonic form as design artifact and for
     verification purposes.
From the canonic form one can choose one of the existing automated
     verification methods for round-based distributed
     algorithms~\syncverif.
     
There are several other frameworks for the verification of
     asynchronous distributed algorithms, e.g.,
     Verdi~\cite{DBLP:conf/pldi/WilcoxWPTWEA15},
     IronFleet~\cite{HawblitzelHKLPR17}, ByMC~\cite{KLVW17:POPL},
     Ivy~\cite{PadonMPSS16}, and Disel~\cite{SergeyWT18}.
Very interesting distributed algorithms have been verified in these
     frameworks.
Still, they require considerable expertise either in manually fitting
     asynchronous code to the fragment that can be dealt with by the method,
     or in guiding interactive theorem provers.
Typically, these works also consider verification of specific
     algorithms which makes it hard to generalize ideas.

Our research belongs to an effort to develop techniques for automated
     reduction to synchronized executions.
Three concurrent approaches in this quest are the exciting results
     in~\cite{BouajjaniEJQ18},~\cite{KraglQH18}
     and~\cite{DBLP:journals/pacmpl/BakstGKJ17,GGB19}.
Compared to their work, our approach is less guided by specific
     communication patterns of existing systems.
Rather we put communication closure in the center of our
     considerations.
Hence, we are more permissive to different communication structures.
For instance, the recent paper~\cite{GGB19} does not allow skipping
     rounds, with the side effect that they cannot model that a
     process remains leader for several consecutive iterations, which
     is an important efficiency mechanism in systems that implement
     ideas from Paxos~\cite{generalizedpaxos} and Viewstamped
     Replication~\cite{OkiL88}.
The notion of k-synchronizability in~\cite{BouajjaniEJQ18} is
     restricted to   FIFO communication channels.
In contrast our method does not make any assumptions  about the
     communication model between processes (works for UDP, TCP/IP).
Moreover, unbounded jumps over phases cannot be captured  by
     k-synchronizability.
The method in~\cite{KraglQH18} adapts verification methods for remote
     procedure calls to leader/follower communication.
As a result they do not support rounds with all-to-all communication
     or that a leader plays also the role of a follower; both is the
     case in our running example.

\section{Conclusion and future work}
\label{sec:conclusion}

We formalized the notion of communication closure of asynchronous protocols and showed that  several
     challenging benchmarks satisfy this property. 
We showed that communication closure captures formally the intuition of protocol designers and is an enabler for a synchronous
     canonic form of asynchronous protocols. 
This canonic form enables the use of different verification
     techniques, and we verified the several benchmarks using the
     Consensus Logic framework~\cite{vmcai}.


We consider the verification of synchronous round-based
     protocols an orthogonal problem, however progress in
     this research area has direct impact on the verification of
     asynchronous protocols that are communication-closed. 
     Roughly the main difficulty regarding automating reasoning of synchronous systems comes from the data 
     they manipulate and not from their control structure. 

Our methods preserves relevant safety and liveness properties.
Reasoning about liveness in \newmodel\ requires assumptions
     about the $HO$ sets, which can be done in Consensus Logic.
However, besides initial theoretical results~\cite{SHQ18}, the  connection
     between $HO$ sets and the asynchronous world is formally not well
     understood, yet.
Thus, there is no automated method that translates asynchronous
     receptions loops with time-outs, etc.\ into $HO$ sets.
This would be required for total correctness regarding liveness  and
     is subject to future work.


%
%
%

\bibliographystyle{splncs04}
\bibliography{biblio}
%




\end{document}